 \documentclass[onefignum,onetabnum]{siamonline190516}

\usepackage{lipsum}
\usepackage{amsfonts}
\usepackage{graphicx}
\usepackage{epstopdf}
\usepackage{algorithmic}
\ifpdf
  \DeclareGraphicsExtensions{.eps,.pdf,.png,.jpg}
\else
  \DeclareGraphicsExtensions{.eps}
\fi

\usepackage{xcolor}
\usepackage[utf8]{inputenc} 
\usepackage[T1]{fontenc}    
\usepackage{hyperref}       
\usepackage{url}            
\usepackage{booktabs}       
\usepackage{nicefrac}       
\usepackage{microtype}      

\usepackage{caption}
\usepackage{subcaption}

\usepackage{etoolbox}
\makeatletter
\patchcmd{\maketitle}
 {\def\@makefnmark}
 {\def\@makefnmark{}\def\useless@macro}
 {}{}
\makeatother

\newtheorem{assumption}{Assumption}

\DeclareMathOperator{\var}{var}

\newcommand{\argmin}{\operatornamewithlimits{argmin}}

\newcommand{\given}{\,|\,}

\def\md{{\mathrm d}}

\def\cA{{\mathcal A}}

\def\cH{{\mathcal H}}
\def\cJ{{\mathcal J}}

\def\bR{{\mathbb R}}

\def\bE{{\mathbb E}}

\def\bN{{\mathbb N}}

\def\NPDF{{\mathcal N}}

\def\cL{{\mathcal L}}

\def\f0{{\mathbf 0}}

\def\KL{{\operatorname{KL}}}

\usepackage{enumitem}
\setlist[enumerate]{leftmargin=.5in}
\setlist[itemize]{leftmargin=.5in}

\newsiamremark{remark}{Remark}
\newsiamremark{hypothesis}{Hypothesis}
\crefname{hypothesis}{Hypothesis}{Hypotheses}
\newsiamthm{claim}{Claim}

\headers{Statistical Finite Elements via Langevin Dynamics}{Akyildiz, Duffin, Sabanis, and Girolami}

\title{Statistical Finite Elements via Langevin Dynamics\thanks{{\"O.~D.~A. was supported by the Lloyd’s Register Foundation Data Centric Engineering Programme and EPSRC Programme Grant EP/R034710/1. C.~D and M.~G were supported by EPSRC grant EP/T000414/1 and M.~G was supported by a Royal Academy of Engineering Research Chair, and EPSRC grants EP/R018413/2, EP/P020720/2, EP/R034710/1, EP/R004889/1. S.S. was supported by the Alan Turing Institute under the EPSRC grant EP/N510129/1.}}}

\author{\"Omer Deniz Akyildiz$^{\star,} $\thanks{The Alan Turing
    Institute, London, NW1 2DB, UK and University of Cambridge, Cambridge, CB2 1PZ, UK}
  \and Connor Duffin$^{\star,} $\thanks{University of Cambridge, Cambridge, CB2 1PZ, UK}
  \and Sotirios Sabanis\thanks{University of Edinburgh, Edinburgh, EH8 9YL, UK
    and The Alan Turing Institute, London, NW1 2DB, UK}
  \and Mark Girolami\thanks{The Alan Turing Institute, London, NW1 2DB, UK and
    University of Cambridge, Cambridge, CB2 1PZ, UK}}

\usepackage{amssymb}
\usepackage{amsopn}

\newcommand{\nobs}{n_{\mathrm{obs}}}
\newcommand{\ninner}{n_{\mathrm{inner}}}

\usepackage[symbol]{footmisc}

\ifpdf
\hypersetup{
  pdftitle={Statistical Finite Elements via Langevin Dynamics},
  pdfauthor={authors}
}
\fi

\begin{document}

\maketitle
\footnotetext[0]{\textsf{$\vphantom{x}^\star$Corresponding authors: \"O.~D.~A
(\email{odakyildiz@turing.ac.uk}) and C.~D (\email{cpd32@cam.ac.uk}).}}

\begin{abstract}
  The recent statistical finite element method (statFEM) provides a coherent
statistical framework to synthesise finite element models with observed
data. Through embedding uncertainty inside of the governing equations, finite
element solutions are updated to give a posterior distribution which quantifies
all sources of uncertainty associated with the model. However to incorporate all
sources of uncertainty, one must integrate over the uncertainty associated with
the model parameters, the known \textit{forward problem} of uncertainty
quantification. In this paper, we make use of Langevin dynamics to solve the
statFEM forward problem, studying the utility of the unadjusted Langevin
algorithm (ULA), a Metropolis-free Markov chain Monte Carlo sampler, to build a
sample-based characterisation of this otherwise intractable measure. Due to the
structure of the statFEM problem, these methods are able to solve the forward
problem without explicit full PDE solves, requiring only sparse matrix-vector products. ULA is
also gradient-based, and hence provides a scalable approach up to high
degrees-of-freedom. Leveraging the theory behind Langevin-based samplers, we
provide theoretical guarantees on sampler performance, demonstrating
convergence, for both the prior and posterior, in the Kullback-Leibler
divergence, and, in Wasserstein-2, with further results on the effect of
preconditioning. Numerical experiments are also provided, for both the prior and
posterior, to demonstrate the efficacy of the sampler, with a \texttt{Python}
package also included.
\end{abstract}



\section{Introduction}
\label{sec:intro}

Uncertainty quantification (UQ) is a fundamental aspect of science, with
increasing attention being cast upon so-called physics-informed statistical
models~\cite{berger2019statistical}. Such models recognise the inherent
uncertainty in model specification, which will almost certainly lead to
misspecification due to incomplete knowledge~\cite{kennedy2001bayesian}. A
particular approach, named statFEM~\cite{girolami2021statistical}, does so via a
probabilistic description
of the model variable $u(x)$, $x \in \Omega \subset \bR^s$, which for the
elliptic PDE case is
\begin{equation}
  \begin{aligned}
    \mathcal{L}_\theta u &= f + \xi, \\
    \log \theta(x) \sim \mathcal{GP}(m(x),
    k_\theta(x, x')), &\quad \xi(x) \sim \mathcal{GP}(0, k_\xi(x, x')).
  \end{aligned}
  \label{eq:statFEMIntro}
\end{equation}
As the mathematical description of the system, this defines the
conditional Gaussian process (GP) $p(u \given \theta) =
\mathcal{GP}([\mathcal{L}_\theta^{-1}f] (x), \mathcal{L}_\theta^{-1} k_\xi(x,
x')\mathcal{L}_\theta^{-*})$, an infinite-dimensional object.
In practice, discretisation is required, which for finite elements proceeds from
the weak form of Equation~\eqref{eq:statFEMIntro},
\[
  \cA_\theta(u, v) = \langle f + \xi, v \rangle, \quad \forall v \in V,
\]
where $\cA_\theta(\cdot, \cdot)$ is the bilinear form generated from
$\cL_\theta$, and $\langle \cdot, \cdot \rangle$ is the appropriate Hilbert
space inner product. Defining the $d$-dimensional basis function expansion
for $u(x) \approx \sum_{i = 1}^d u_i \phi_i(x)$,
$v(x) \approx \sum_{i = 1}^d v_i \phi_i(x)$ then gives the induced Gaussian
\[
  p(u \given \theta) = \NPDF(A_\theta^{-1} b, A_\theta^{-1} G A_\theta^{-\top}),
\]
where $A_{\theta, ij} = \cA_\theta(\phi_i, \phi_j)$, $b_j =
\langle f, \phi_j \rangle$, $G_{ij} = \langle \phi_i, \langle k_\xi(\cdot,
\cdot), \phi_j \rangle \rangle$.
The now finite-dimensional probability distribution provides a quantification of
the \textit{a priori} belief in the model specification and serves as a
reference measure for further inference and data assimilation;
observed data, $y \in \bR^{d_y}$, can now be incorporated to give the posterior
over the FEM coefficients $u \in \bR^d$.
Examples~\cite{duffin2021lowrank,duffin2021statistical,febrianto2021selfsensing}
have demonstrated the usefulness of this approach, with the computed posterior
distribution providing a statistically coherent synthesis of physics and data,
with an interpretable UQ.

Yet to provide UQ, uncertainty associated with $\theta$ must be taken into
account in order to compute the statFEM prior and posterior distributions. This
requires marginalising over $\theta$; for the prior this is $p(u) = \int p(u
\given \theta) p(\theta) \md \theta$. This integral is intractable, and some
sort of approximation is required. In previous work, $\theta$ was either fixed
or a first-order expansion was taken, to give an approximate marginal measure.
In this work we take a fully Bayesian approach and marginalise over these
parameters directly.

Due to stochasticity, a Monte-Carlo (MC) approximation can be used. For
elliptic stochastic PDEs, one such approximation is multilevel MC
(MLMC)~\cite{cliffe2011multilevel,giles2015multilevel}, in which the system is
solved on a hierarchy of more computationally expensive models, to estimate a
possibly nonlinear functional of interest. Additional methods that
involve MC sampling for functional estimation are multilevel Markov chain Monte
Carlo (MLMCMC)~\cite{dodwell2015hierarchical} and
quasi-MC~\cite{graham2011quasimonte}. However, in this work we target the
measure defined by the finite element solution, both \textit{a priori} and
\textit{a posteriori}, an inherently high-dimensional probability distribution.
Thus in order to estimate $p(u)$ and $p(u \given y)$, we require samplers
that are robust to mesh-refinement and have been applied to physics-based
problems. This has been studied in Bayesian inversion~\cite{stuart2010inverse}.

For high-dimensional sampling in the inverse problems context,
Metropolis-Hastings (MH) Markov chain Monte Carlo (MCMC) samplers are often
used~\cite{robert2004monte}, one example of which is the preconditioned
Crank-Nicolson (pCN) algorithm~\cite{cotter2013mcmc}. The pCN algorithm is
defined on the function space of the underlying problem, before discretisation,
ensuring robustness under mesh-refinement. The standard proposal mechanism uses
a random walk, with the covariance being the same as the prior distribution.
Extensions of pCN, to alternative proposals that make use of the geometry of the
target measure, are also considered
in~\cite{beskos2017geometric,law2014proposals}. An alternate approach is the
stochastic Newton method~\cite{martin2012stochastic}, which constructs a proposal mechanism from
the Hessian of the log-density.

However, in high-dimensional scenarios, recent works have shown that taking an
MH-free approach can also lead to efficient samplers that scale well with
increasing dimension. Inspired by the success of gradient-based optimisers in
machine learning, methods based on Langevin diffusions, such as the unadjusted
Langevin algorithm (ULA)
\cite{dalalyan2019user,durmus2017nonasymptotic,durmus2019high,vempala2019rapid,wibisono2018sampling}
have received significant attention. Our approach is further
motivated by the standard ULA iteration, targeting
$p(u \given \theta)$, being
\[
  u_k = u_{k - 1} - \eta A_\theta^\top G^{-1} (A_\theta u_{k - 1} - b)
  + \sqrt{2 \eta} Z_k, \quad Z_k \sim \NPDF(0, I),
\]
which does not require the inversion of the FEM stiffness matrix. For
judicious choice of the covariance structure of $G$, this gives an iterative
sampling approach which can characterise the measure $p(u \given \theta)$
through sparse matrix-vector products, solving the forward problem without ever
needing to solve the model itself.

The contribution of this work is to establish a link between popular
gradient-based sampling methods and statFEM, paving the way for their general
use in stochastic elliptic PDEs. These schemes are based on Euler
discretisations of continuous-time Langevin diffusions, and, due to the
structure of statFEM, assumptions required for these algorithms to perform well
are satisfied a priori: they are ideal tools to sample the induced
measures. After defining the problem in Section~\ref{sec:background}, we define
various samplers and provide explicit convergence rates in
both the KL divergence and in Wasserstein-2, for both the prior and posterior
measures, in Section~\ref{sec:Theory}. We also study the effect of
preconditioning on the
allowed step-sizes of the algorithm. In Section~\ref{sec:experiments}, we
empirically demonstrate these results on case studies for
the Poisson equation both \textit{a priori} and \textit{a posteriori}. We also
include a \texttt{Python} package to accompany this paper, containing all
implementations\footnote{Available
at~\url{https://www.github.com/connor-duffin/ula-statfem}}.

This paper enables the use of similar samplers in this framework, such as
variance-reduced, underdamped, decentralised, and parallelised Langevin
algorithms. All of these schemes can be analysed and incorporated by extending
the techniques we provide in this paper and can contribute to the use of the
statFEM methodology under various scenarios where other samplers may be
preferable.

\subsection*{Notation} For a set $A$, we denote its boundary as $\partial A$. We denote the space of functions which are compactly supported and infinitely differentiable with $\mathcal{C}_c^\infty$. We say that $v: \Omega \to \bR$ is a weak derivative of $u: \Omega\to\bR$ in the direction of the $i$th coordinate if
\begin{align*}
    \int_\Omega u(x) \varphi(x) \md x = - \int_\Omega v(x) \frac{\partial \varphi}{\md x_i}(x) \md x,
\end{align*}
for all compactly supported functions $\varphi \in \mathcal{C}_c^\infty$. The Sobolev space $H^1(\Omega)$ is defined as a set of all functions that are square integrable and have weak derivatives in all coordinate directions. We also define the Sobolev space $H^1_0(\Omega)$ to denote the functions $u \in H^1(\Omega)$ and identically zero on $\partial \Omega$. The notation $L^2(\Omega)$ denotes the space of square-integrable functions and $\langle f, g\rangle = \int_\Omega f(x) g(x) \md x$ the associated inner product w.r.t. the Lebesgue measure.

\section{Statistical Finite Elements for Linear PDEs}
\label{sec:background}
Consider the elliptic example of Section~\ref{sec:intro},
Equation~\eqref{eq:statFEMIntro}, with Dirichlet boundary conditions
\begin{equation}
  \label{eq:FormulationPDE}
  \begin{aligned}
    \mathcal{L}_\theta u &= f + \xi, \quad x \in \Omega, \\
    u &= 0, \quad \quad x \in \partial \Omega,
  \end{aligned}
\end{equation}
where $u := u(x)$, $f = f(x)$,  $x \in \Omega \subset \bR^d$. We note that
hyperparameters of the Gaussian processes $\xi$, $\theta$ are assumed to be known.

A Gaussian prior measure can be constructed by considering the
weak solutions of \eqref{eq:FormulationPDE}, conditioned on $\theta$.
The weak form is specified from multiplying both sides of the equation with a
test function $v \in H_0^1(\Omega)$, assuming that $u \in H_0^1(\Omega)$,
$f, \xi \in L^2(\Omega)$, and $\theta \in L^\infty(\Omega)$, and integrating over
$\Omega$,
\begin{align}
    \mathcal{A}_\theta(u, \psi) = \langle f, \psi\rangle + \langle \xi, \psi\rangle.
\end{align}
The space $H_0^1(\Omega)$ admits an the orthonormal basis
$\{\phi_i\}_{i\in \bN}$, so for $u \in H^1_0(\Omega)$,
$u = \sum_{i=1}^\infty u_i \phi_i(x)$ and hence
\begin{align*}
  \sum_{i=1}^\infty u_i \mathcal{A}_\theta(\phi_i, \phi_j) = \langle f, \phi_j\rangle + \langle \xi, \phi_j\rangle,
  \quad j \in \bN,
\end{align*}
as $v \in H_0^1(\Omega)$. Therefore, $\mathcal{A}_\theta$ can be seen as
an infinite dimensional matrix $A_\theta$ with entries $A_{\theta, ij} =
\mathcal{A}_\theta(\phi_i, \phi_j)$.

To construct the finite-dimensional approximation, we first subdivide the domain
$\Omega$ to construct the triangulation $\Omega_h \subset
\Omega$, with vertices $\{x_j\}_{j = 1}^d$, where the maximal length of the
sides of the triangulation is given by $h$. The polynomial basis functions
$\{\phi_j\}_{j = 1}^d$ are then defined on the mesh, having the property that
$\phi_i(x_j) = \delta_{ij}$. In this work we consider the linear polynomial
$C^1(\Omega)$ basis functions only, and denote by $V_h \subset H_0^1(\Omega)$ as
the span of these basis functions. Projecting from the infinite-dimensional
space gives the finite-dimensional approximation $u(x) \approx u_h(x) = \sum_{j
  = 1}^d u^{(j)} \phi_j(x)$. The vector of FEM coefficients is
$u = (u^{(1)}, \ldots, u^{(d)})$, reusing the notation from
\eqref{eq:FormulationPDE}; dimensionality should be clear from context.

The finite-dimensional weak form induces a Gaussian distribution
over the FEM coefficients $u$. Conditioned on $\theta$, this is
\begin{align} \label{eq:ConditionalDistribution}
  p(u|\theta) = \NPDF(A_\theta^{-1}b, A^{-1}_\theta G A_\theta^{-\top}),
\end{align}
where $A_{\theta,ij} = \mathcal{A}_\theta(\phi_i, \phi_j)$,
$b_j = \langle f, \phi_j\rangle$,
and $G_{ij} = \langle \phi_i, \langle k_\xi(\cdot, \cdot), \phi_j \rangle\rangle$. The marginal distribution of the
solution $u$ is
\begin{align}\label{eq:MarginalDistributionOfu}
    p(u) = \int p(u | \theta) p(\theta) \md \theta,
\end{align}
and the first goal in this paper is to describe a methodology for sampling from this marginal. This is done in Section~\ref{sec:SamplingPrior}, by targeting the joint $p(u,\theta) = p(u|\theta)
p(\theta)$ and considering the marginal of $u$. Our second goal is to sample from
the posterior distribution
\begin{align}
    p(u | y) = \frac{p(u, y)}{p(y)}
\end{align}
To solve this problem, we construct samplers in the extended space $p(u, \theta | y)$ and target its $u$-marginal. The inference of the solution amounts to a Bayesian inference procedure and we tackle this problem by developing a sampler for the posterior measure over the numerical solutions of the PDE under consideration, in Section~\ref{sec:SamplingPosterior}.

We note that a classical approach to sample from the joint $p(u, \theta)$ is to
consider a Gibbs sampler which targets $p(u \given \theta)$ and $p(\theta \given
u)$ in turn. However, as $p(\theta \given u)$ is not of the form of any known
distribution (whilst its prior is a GP), we choose to sample from the
marginal through exact sampling from $p(\theta)$, which can be done efficiently
using, for example, Kronecker methods~\cite{saatci2011scalable} or circulant embedding~\cite{dietrich1997fast}. We also note
that this approach is useful in scenarios in which samples from
$\theta$ have already been obtained, from some distribution of
interest, and their effect on the model solutions is desired.

\subsection{Conditional Langevin SDEs}
In this section, we introduce conditional Langevin SDEs to sample from
$p(u|\theta)$ and $p(u | y, \theta)$ for fixed $\theta$. In particular, we consider
\begin{align}\label{eq:statFEM_SDE}
    \md u_t = - \nabla \Phi_\theta(u_t) \md t + \sqrt{2} \md B_t,
\end{align}
where $(B_t)_{t \geq 0}$ is a Brownian motion and
\begin{align*}
  \Phi_\theta(u) = \frac{1}{2}(A_\theta u - b)^\top G^{-1}(A_\theta u - b).
\end{align*}
By construction, this SDE targets the conditional distribution given in Eq.~\eqref{eq:ConditionalDistribution}. The potential advantage of using such a diffusion can already be observed by the expression of the diffusion given in \eqref{eq:statFEM_SDE}. In particular, note that the drift in this diffusion is given by
\begin{align}
    \nabla_u \Phi_\theta(u) = A_\theta^\top G^{-1} A_\theta u - A_\theta^\top G^{-1} b.
\end{align}
This expression does not contain any inverse of the form $A_\theta^{-1}$ which implies an efficient computational scheme. Note that here $G^{-1}$ is a matrix that is easy to invert, e.g., a diagonal matrix.

Before going into an efficient discretisation of the conditional Langevin diffusion, we note that the SDE in \eqref{eq:statFEM_SDE} has an analytical solution as summarised in the following remark.
\begin{remark} The Langevin SDE given in \eqref{eq:statFEM_SDE} has the following analytical solution (see, e.g., \cite{wibisono2018sampling})
\begin{align}
    u_t \overset{d}{=} A_{\theta}^{-1} b + e^{-t A_\theta^\top G^{-1} A_\theta} (u_0 - A_\theta^{-1} b) + (A_\theta^{-1} G^{} A_\theta^{-\top})^{1/2} \left(I - e^{-2t A_\theta^\top G^{-1} A_\theta}\right)^{1/2} Z
\end{align}
where $Z \sim \NPDF(0,I)$ is independent of $u_0$.
\end{remark}
It is clear that having an analytical solution is not useful in this example, since it necessitates computing $A_\theta^{-1}$ and other extra computations, which ends up being impractical compared to performing an inversion in the first place. However, despite the exact solution requiring heavy computations, the associated numerical scheme to the Langevin SDE does not require the same computational resources, and can be formally shown to track the underlying diffusion closely.

To sample from the conditional SDE targeting $p(u|\theta) \propto \exp(-\Phi_\theta(u))$, we consider the unadjusted Langevin algorithm (ULA) which is the standard Euler-Maruyama discretisation of the SDE \eqref{eq:statFEM_SDE} and given as
\begin{align}
u_{k+1} &= u_k - \eta \nabla_u \Phi_\theta(u_k) + \sqrt{2\eta} Z_{k+1},\nonumber \\
&= u_k - \eta A_\theta^\top G^{-1} A_\theta u_k + \eta A_\theta^\top G^{-1} b + \sqrt{2\eta} Z_{k+1},\label{eq:ULAiteration}
\end{align}
where $(Z_k)_{k\geq 0}$ are i.i.d standard Normal random variables, $Z_k \sim \NPDF(0, I)$ for every $k$. Given that this scheme consists of a noisy linear mapping, we can also write down the asymptotic distribution of the ULA exactly. To see this, first rewrite the iterations as
\begin{align*}
u_{k+1} - A_\theta^{-1} b = (I - \eta A_\theta^\top G^{-1} A_\theta)(u_k - A_\theta^{-1} b) + \sqrt{2\eta} Z_{k+1}.
\end{align*}
then limiting distribution of the conditional ULA is \cite{wibisono2018sampling}
\begin{align*}
p_\infty^\eta = \NPDF\left(A_\theta^{-1} b, A_\theta^{-1} G A_\theta^{-\top} \left( I - \frac{\eta}{2} A_\theta^\top G^{-1} A_\theta \right)^{-1} \right).
\end{align*}
It is interesting to observe that, in this case, the mean estimates computed using the stationary measure would be unbiased. The discretisation scheme results in the bias of the uncertainty estimates, which can be made arbitrarily small as $\eta \to 0$.

To sample from the conditional posterior distribution $p(u | y, \theta)$, we define the potential
\begin{align*}
    \Phi^y_\theta(u) &= -\log p(u | y, \theta) \\
    &= -\log p(y|u) - \log p(u | \theta).
\end{align*}
In this case, the Langevin SDE targeting $p(u | y,\theta)$ takes the form
\begin{align*}
    u_{k+1} = u_k - \eta \nabla_u \Phi_\theta^y(u_k) + \sqrt{2\eta} Z_{k+1}.
\end{align*}
Note that this scheme is general and can handle nonlinear observation models
unlike an exact sampling method.

\subsection{Unadjusted Langevin algorithm for sampling the prior measure}\label{sec:SamplingPrior}
The first problem we are interested in solving is to sample from the marginal prior of the solutions $p(u)$ induced by the statistical FEM construction. To solve this problem, we design a sampler that targets the joint distribution $p(u,\theta)$ and we obtain the marginal $p(u)$ from these samples. In order to sample from the joint $p(u,\theta)$, we leverage the conditional schemes and consider the following scheme
\begin{align}
    \theta &\sim p(\theta), \label{eq:SampleTheta}\\
    u_{k+1} &= u_k - \eta \nabla_u \Phi_\theta(u_k) + \sqrt{2\eta} Z_{k+1}.\label{eq:SampleU}
\end{align}
We show in Sec.~\ref{sec:Theory} that this scheme can indeed be used for approximately sampling from the marginal prior $p(u)$ and provide some theoretical guarantees for this sampler.
\subsection{Unadjusted Langevin algorithm for sampling the posterior measure}\label{sec:SamplingPosterior}
Sampling from $p(u)$ is not our final goal, as we are interested in sampling from the posterior measure $p(u|y)$. To this end, we consider the sampler
\begin{align}
    \theta &\sim p(\theta), \label{eq:SampleThetaData}\\
    u_{k+1} &= u_k - \eta \nabla_u \Phi_\theta^y(u_k) + \sqrt{2\eta} Z_{k+1}.\label{eq:SampleUData}
\end{align}
which samples from $p(u,\theta | y)$ (See Section~\ref{sec:Theory}). We note that
\begin{align*}
    \nabla \Phi_\theta^y(u) = - \nabla \log p(u | y, \theta) = -\nabla \log p(y|u) - \nabla \log p(u|\theta).
\end{align*}
Therefore the implementation of the recursions \eqref{eq:SampleThetaData}--\eqref{eq:SampleUData} are straightforward given $\theta$.

\subsection{Preconditioned unadjusted Langevin schemes}
For badly conditioned problems, preconditioning can improve the convergence and stability of the above ULA schemes significantly. Therefore, to sample from the prior measure $p(u)$, we consider the preconditioned sampler
\begin{align}
    \theta &\sim p(\theta) \label{eq:preconditionedMarginalSamplerPrior} \\
    u_{k+1} &= u_k - \eta M \nabla \Phi_\theta(u_k) + \sqrt{2\eta} M^{1/2} W_{k+1},\label{eq:preconditionedMarginalSamplerConditional},
\end{align}
where $(W_k)_{k\geq 0}$ is the sequence of standard Normal random variables. Similarly, in order to sample from the posterior $p(u | y)$, we consider
\begin{align}
    \theta &\sim p(\theta) \label{eq:preconditionedPosteriorSamplerPrior} \\
    u_{k+1} &= u_k - \eta M \nabla \Phi_\theta^y (u_k) + \sqrt{2\eta} M^{1/2} W_{k+1},\label{eq:preconditionedPosteriorSamplerConditional}.
\end{align}
We provide theoretical analysis of the preconditioned schemes in Section~\ref{sec:analysis:preconditioned} and demonstrate their utility in Section~\ref{sec:experiments}.

\subsection{The Algorithm} To draw an approximate sample from the marginal distributions $p(u)$ or $p(u
\given y)$, using ULA, the approach we take in this paper is to first draw
$\theta_k \sim p(\theta)$, and then run a sub-chain for $p(u \given \theta_k)$
(or $p(u \given \theta_k, y)$), for $\ninner$ iterations. This sub-chain is
initialised to the previous iterate $u_{k - 1}$. Due to the warm start, it is
assumed that $\ninner$ is small; in Section~\ref{sec:experiments} we take
$\ninner = O(10)$. Note also that the inner iterations are cheap as there is no
requirement for FEM assembly, requiring only matrix-vector products with the
sparse $A_\theta$. This gives the set of samples $\{u_{i, k}\}_{i = 1}^{\ninner}
\sim p(u \given \theta_k)$, and we take the joint sample as $(u_k, \theta_k) =
(u_{\ninner, k}, \theta_k) \sim p(u, \theta)$ (the full algorithm is shown in
Algorithm~\ref{alg:pula}). The inner iterations ensure that the sample is taken
from the target measure, whilst also decorrelating the $u_k$ samples.

\begin{algorithm}
  \caption{StatFEM ULA sampler.}
  \label{alg:pula}
  \begin{algorithmic}
    \STATE{Let $M \in \bR^{d \times d}$ be the preconditioner, $u_0 \in \bR^d$
      the initial condition.}
    \FOR{$k = 1, \ldots, K$}
      \STATE{$\theta_k \sim p(\theta)$.}
      \STATE{$u_{0, k} = u_{k - 1}$.}
      \FOR{$i = 1, \ldots, \ninner$}
        \STATE{$Z_{i, k} \sim \mathcal{N}(0, I)$.}
        \STATE{$u_{i, k} = u_{i - 1, k}
          - \eta M \nabla \Phi_\theta (u_{i - 1, k})
          + \sqrt{2\eta} M^{1/2} Z_{i, k}$}
      \ENDFOR
      \STATE{$u_k = u_{\ninner, k}$}.
    \ENDFOR
    \RETURN $\{(u_k, \theta_k)\}_{k = 1}^K$
  \end{algorithmic}
\end{algorithm}

\section{Analysis}\label{sec:Theory}
In this section, we analyse the proposed schemes and prove convergence rates. In particular, in Section~\ref{sec:analysis:conditional}, we prove the convergence of the conditional ULA chains, where the $(u_k)_{k\geq 0}$ iterates run with fixed $\theta$. Then, in Section~\ref{sec:analysis:prior}, we start by analysing the sampler defined in Eqs.~\eqref{eq:SampleTheta}--\eqref{eq:SampleU}. This sampler aims at sampling from the marginal $p(u)$. Next, in Section~\ref{sec:analysis:posterior}, we analyse the sampler defined in Eqs.~\eqref{eq:SampleThetaData}--\eqref{eq:SampleUData}. This analysis follows from similar arguments of the analysis of the marginal sampler of $p(u)$. Finally, in Section~\ref{sec:analysis:preconditioned}, we extend these results to the case when the preconditioners are used in the ULA chains. We demonstrate, quantitatively, that a good preconditioner choice can impact and improve the error and convergence rates significantly, an observation we verify in the experimental section.

Throughout this section, we assume the following.
\begin{assumption}\label{assmp:LowerUpperBounds} We assume there exists $m > 0$ and $L > 0$ such that
\begin{align}
0 < m = \inf_{\theta} \lambda_{\min} (A_\theta^\top G^{-1} A_\theta)  \quad \quad \textnormal{and} \quad \quad L = \sup_\theta \lambda_{\max}(A_\theta^\top G^{-1} A_\theta) < \infty.
\end{align}
\end{assumption}
This assumption implies that, for every $\theta$, we have
\begin{align*}
m I \preceq \nabla^2 \Phi_\theta(u) \preceq L I.
\end{align*}
Using this definition, we also define \textit{the worst-case condition number}
\begin{align}\label{eq:KappaMax}
    \kappa_{\max} := \frac{L}{m} < \infty.
\end{align}
The existence and finiteness of these quantities is verifiable under some conditions.

\subsection{The analysis in KL divergence for the conditional ULA}\label{sec:analysis:conditional}
Two standard conditions for analysing the convergence of Langevin-type scheme is strong convexity and gradients of the potential being Lipschitz. In our case, we can prove that these conditions hold as follows.
\begin{lemma}\label{lem:ConvexSmooth} Given Assumption~\ref{assmp:LowerUpperBounds}, the gradient $\nabla \Phi_\theta(u)$ is Lipschitz, i.e.,
\begin{align*}
    \|\nabla \Phi_\theta(u) - \nabla \Phi_\theta(u')\| \leq L_\theta \|u - u'\|,
\end{align*}
where $L_\theta = \lambda_{\max}(A_\theta^\top G^{-1} A_\theta)$. Moreover, for fixed $\theta$, the function $\Phi_\theta(\cdot)$ is strongly convex with $m_\theta = \lambda_{\min}(A_\theta^\top G^{-1} A_\theta)  > 0$ which follows from Assumption~\ref{assmp:LowerUpperBounds}.
\end{lemma}
\begin{proof}
The first claim follows from
\begin{align*}
\|\nabla_u \Phi_\theta(u) - \nabla_u \Phi_\theta(u')\| &= \|A_\theta^\top G^{-1} A_\theta u - A_\theta^\top G^{-1} A_\theta u'\|, \\
&\leq L_\theta \|u-u'\|,
\end{align*}
where the last line is obtained by the properties of the matrix norm. For proving the second claim, note that a function is strongly convex iff $\nabla^2 \Phi_\theta(u) \succeq m_\theta I$. In our case, $\nabla^2 \Phi_\theta(u) = A_\theta^\top G^{-1} A_\theta$, hence $m_\theta = \lambda_{\min}(A_\theta^\top G^{-1} A_\theta) > 0$.
\end{proof}
Assumption~\ref{assmp:LowerUpperBounds} states that $L = \sup_\theta L_\theta < \infty$ and $m = \inf_\theta m_\theta >0$.

Next, we provide the convergence rate of the conditional ULA, adapted from \cite{vempala2019rapid}.
\begin{theorem}\label{thm:conditionalConv} Under Assumption~\ref{assmp:LowerUpperBounds} and step-size given by $0 < \eta \leq \frac{m_\theta}{4 L_\theta^2}$, we have
\begin{align*}
\KL(p_k(u|\theta) || p(u|\theta)) \leq e^{-m_\theta \eta k} \KL(p_0(u) || p(u|\theta)) + {8 \eta d L_\theta \kappa_\theta},
\end{align*}
where $\kappa_\theta$ is the condition number of the matrix $A_\theta^\top G^{-1} A_\theta$.
\end{theorem}
\begin{proof}
Using Theorem~2 in \cite{vempala2019rapid}, we directly obtain
\begin{align*}
\KL(p_k(u|\theta), p(u|\theta)) \leq e^{-m_\theta \eta k} \KL(p_0(u), p(u|\theta)) + \frac{8 \eta d L_\theta^2}{m_\theta}.
\end{align*}
Observing that $\kappa_\theta = L_\theta / m_\theta$ is the condition number of the matrix $A_\theta^\top G^{-1} A_\theta$, we obtain the result.
\end{proof}
This result is for fixed $\theta$, but still is insightful to demonstrate the dependence of the error bound to the condition number of the matrix $A_\theta^\top G^{-1} A_\theta$. This bound readily implies that we have to choose small step-sizes for badly conditioned problems. Furthermore, it also motivates preconditioning to improve this condition number, as we will discuss in Sec.~\ref{sec:analysis:preconditioned}.

\subsection{The analysis in KL divergence for the ULA prior sampler}\label{sec:analysis:prior}
In this section, we prove theoretical guarantees for the sampler defined in \eqref{eq:SampleTheta}--\eqref{eq:SampleU}.

Similarly to the conditional case, we can still attain a similar result for the marginal sampler of $u$ as the conditional case as we outline below.
\begin{theorem}\label{thm:MarginalConvergence} Under Assumption~\ref{assmp:LowerUpperBounds} and $0 < \eta \leq \frac{m}{4 L^2}$,
\begin{align*}
    \KL(p_k(u) || p(u)) \leq e^{-m \eta k} \bE\left[ \KL(p_0(u) || p(u|\theta))\right] + 8 \eta L d \kappa_{\max},
\end{align*}
where $\kappa_{\max}$ is given in \eqref{eq:KappaMax}.
\end{theorem}
\begin{proof}
This result is a straightforward consequence of Thm.~\ref{thm:conditionalConv}. In particular, taking expectations of the bound provided in Thm.~\ref{thm:conditionalConv} w.r.t. $\theta$, we obtain for the l.h.s.
\begin{align*}
    \int \KL(p_k(u|\theta) || p(u | \theta)) p(\theta) \md \theta &= \int \int \log \frac{p_k(u | \theta)}{p(u|\theta)} p_k(u | \theta) p(\theta) \md u \md \theta, \\
    &= \int \int \log \frac{p_k(u | \theta) p(\theta)}{p(u|\theta) p(\theta)} p_k(u | \theta) p(\theta) \md u \md \theta, \\
    &= \KL(p_k(u,\theta) || p(u,\theta)).
\end{align*}
Then, Thm.~\ref{thm:conditionalConv} implies that
\begin{align*}
    \KL(p_k(u, \theta) || p(u, \theta)) \leq e^{-m \eta k} \bE[\KL(p_0(u) || p(u | \theta))] + 8 \eta L d \kappa_{\max},
\end{align*}
by observing that $L_\theta < L$ and $m < m_\theta$ for every $\theta$ and $\kappa_{\max} = L / m$. Finally the chain rule of KL-divergence (see Lemma~\ref{lem:ChainRuleKL}) implies that
\begin{align*}
    \KL(p_k(u) || p(u)) \leq \KL(p_k(u,\theta) || p(u,\theta)),
\end{align*}
which concludes the proof.
\end{proof}
These guarantees, based on the KL divergence, can be extended to the Wasserstein-2 distance. This is summarised in the following proposition.
\begin{proposition} Let Assumption~\ref{assmp:LowerUpperBounds} hold and $\eta \leq 2 / (m + L)$. Then, we can obtain
\begin{align}
    W_2(p_k(u), p(u)) \leq \sqrt{2} (1 - m\eta)^{k} \bE[W_2(p_0(u), p(u|\theta))] + \frac{7}{3} \kappa (\eta d)^{1/2}.
\end{align}
\end{proposition}
\begin{proof}
Note that, we can write for $\eta \leq 2 / (m+L)$ that
\begin{align*}
    W_2^2(p_k(u | \theta), p(u|\theta)) \leq 2 (1 - m\eta)^{2k} W_2^2(p_0(u), p(u|\theta)) + \frac{49}{9} \kappa^2 \eta d
\end{align*}
as a consequence of Theorem~1 in \cite{dalalyan2019user} and $(a + b)^2 \leq 2 a^2 + 2 b^2$. Now using Lemma~\ref{lem:conditionalW2}, we can obtain
\begin{align*}
    W_2^2(p_k(u), p(u)) \leq 2 (1 - m\eta)^{2k} \bE[W_2^2(p_0(u), p(u|\theta))] + \frac{49}{9} \kappa^2 \eta d.
\end{align*}
Using $\sqrt{x+y} \leq \sqrt{x} + \sqrt{y}$ and the Jensen's inequality, we obtain the claimed result.
\end{proof}
A similar result was provided for mixtures of log-concave distributions in \cite[Theorem~3]{dalalyan2019user}.

\subsection{The analysis in KL divergence for the ULA sampler for the posterior}\label{sec:analysis:posterior}
In this section, we analyse the sampler \eqref{eq:SampleThetaData}--\eqref{eq:SampleUData}. We first note that the distribution $p(y|u)$ denotes observation model. Next, in order to define the sampler, we define the potential
\begin{align}
    \Phi_\theta^y(u) &= -\log p(y | u) - \log p(u|\theta).
\end{align}
Note that, we allow here for general observation models $p(y | u)$ which can be potentially nonlinear. Then we assume the following lower and upper bounds.
\begin{assumption}\label{assmp:PosteriorLowerUpperBounds} We assume there exists $m_y > 0$ and $L_y > 0$ such that
\begin{align}
0 < m_y = \inf_{\theta, u} \lambda_{\min} (\nabla^2 \Phi_\theta^y(u))  \quad \quad \textnormal{and} \quad \quad L_y = \sup_{\theta, u} \lambda_{\max}(\nabla^2 \Phi_\theta^y(u)) < \infty.
\end{align}
In other words,
\begin{align*}
    m_y I \preceq \nabla^2 \Phi^y_\theta(u) \preceq L_y I.
\end{align*}
for every $u \in \bR^d$.
\end{assumption}
Similar to the previous section, we define the worst-case condition number related to the posterior as
\begin{align}\label{eq:KappaMaxy}
    \kappa^y_{\max} = \frac{L_y}{m_y}.
\end{align}
\begin{remark} (Non-log-concave likelihoods) Recall that
\begin{align*}
    \Phi_\theta^y(u) = \Phi_\theta(u) - \log p(y | u).
\end{align*}
Since we know that $\Phi_\theta$ already satisfies Assumption~\ref{assmp:LowerUpperBounds}, this allows for $\log p(y | u)$ to be non-log-concave, as long as $\Phi_\theta^y$ is still strongly convex. In other words, if the problem is well-conditioned, we can handle non-log-concave nonlinear observation models. In particular, if $\Phi_\theta(u)$ is $m$-strongly convex, then we can afford $-\log p(y | u)$ to be $(m - m_y)$-weakly-convex (for $m_y < m$) in the sense of \cite{vial1983strong}\footnote{A function $f(x)$ is $\rho$-weakly-convex if
\begin{align*}
    f(x) + \frac{s}{2}\|x\|^2_2
\end{align*}
is convex for $s \geq \rho$.} so that $\Phi_\theta^y(u)$ is $m_y$-strongly-convex. Note that, this also demonstrates that the lack of strong log-concavity of $p(y|u)$ results in a slower convergence rate and a worse condition number.
\end{remark}
\begin{remark} (Non-log-concave likelihoods II) We further note that all results we have derived hold under a more general assumption on target distributions, namely the Log-Sobolev inequality (LSI) \cite{vempala2019rapid} (which includes a family of non-log-concave likelihoods). This implies that, given that $\Phi_\theta(u)$ is strongly-convex, this allows $-\log p(y | u)$ to be non-convex (even weaker than weak-convexity), as long as the posterior $p(u|y)$ satisfies the LSI. This may account for a large family of likelihoods in practice.
\end{remark}

In the linear case, following \cite{girolami2021statistical}, we define the observation model as
\begin{align}
    y = H u + e,
\end{align}
where $H u$ is the projected finite element solution and $e \sim \NPDF(0,R)$. We assume that $e$ is a $d_y$-dimensional zero-mean Gaussian random variable with covariance $R = \sigma_e^2 I$. Therefore the conditional $p(y|u)$ can be written explicitly as
\begin{align}
    p(y|u) = \NPDF(H u, R).
\end{align}
In this case, we obtain
\begin{align}
    \Phi_\theta^y(u) &= -\log \NPDF(y; H u,  R) - \log \NPDF(u; A_\theta^{-1} b, A_\theta^{-\top} G A_\theta^{-1}).
\end{align}
We then have
\begin{align*}
    \nabla_u \Phi^y_\theta(u) &= - H^\top R^{-1} (y - H u) + A_\theta^\top G^{-1} A_\theta u - A_\theta^\top G^{-1} b, \\
    &= (A_\theta^\top G^{-1} A_\theta + H^\top R^{-1} H) u - H^\top R^{-1} y - A_\theta^\top G^{-1} b.
\end{align*}
Note that
\begin{align*}
    \nabla^2 \Phi^y_\theta(u) = A_\theta^\top G^{-1} A_\theta + H^\top R^{-1} H.
\end{align*}
The meaning of Assumption~\ref{assmp:PosteriorLowerUpperBounds} in this case is that the contribution from the matrix $H^\top R^{-1} H$ should not make the condition number worse. In this case, it can be verified that this term has positive eigenvalues, therefore Assumption~\ref{assmp:PosteriorLowerUpperBounds} holds in practical settings. This assumption in particular implies that $\Phi_\theta^y$ is $L_y$-smooth and $m_y$-strongly convex.

Next, we state our result.
\begin{theorem} Let Assumption~\ref{assmp:PosteriorLowerUpperBounds} hold and $0 < \eta \leq \frac{m_y}{L_y^2}$. Then, we have
\begin{align*}
\KL(p_k(u|\theta, y) || p(u|\theta, y)) \leq e^{-m_y \eta k} \KL(p_0(u) || p(u|y, \theta)) + {8 \eta d L_y \kappa^y_{\max}}.
\end{align*}
\end{theorem}
\begin{proof}
The proof is identical to the proof of Theorem~\ref{thm:conditionalConv} since $p(\cdot | \theta, y)$ is $m^y$-log-strongly concave and $L^y$ log-smooth. 
\end{proof}
A similar argument as in the previous section lets us to obtain our final result about the posterior measure $p(u | y)$.

\begin{theorem}\label{thm:MarginalPosteriorConv} Under Assumption~\ref{assmp:PosteriorLowerUpperBounds} and $0 < \eta \leq \frac{m_y}{4 L_y^2}$,
\begin{align*}
    \KL(p_k(u|y) || p(u|y)) \leq e^{-m_y \eta k} \bE\left[ \KL(p_0(u) || p(u|\theta, y))\right] + 8 \eta L_y d \kappa^y_{\max},
\end{align*}
where $\kappa^y_{\max}$ is given in \eqref{eq:KappaMaxy}.
\end{theorem}
\begin{proof}
The proof is straightforward and follows the same steps as Theorem~\ref{thm:MarginalConvergence} by replacing the notation for the prior with the posterior.
\end{proof}
\subsection{The analysis of the preconditioned Langevin schemes}\label{sec:analysis:preconditioned}
To improve convergence, we also consider the preconditioned ULA. We start with the conditional case, targeting $p(u|\theta)$, whose Langevin diffusion has the form
\begin{align}\label{eq:PreconditionedLangevin}
\md u_t = - M \nabla \Phi_\theta(u_t) \md t + \sqrt{2} M^{1/2} \md B_t.
\end{align}
The strategy used in this part of the analysis is to rewrite this diffusion as a standard overdamped diffusion with a variable transformation (see, e.g., \cite{alrachid2018some}). This allows us to utilise the results for the standard conditional diffusion and obtaining the same rates by using the fact that the KL divergence is invariant w.r.t. invertible transformations. Given the convergence results of the conditional case, it is straightforward to extend our results for the marginal sampler defined in \eqref{eq:preconditionedMarginalSamplerPrior}--\eqref{eq:preconditionedMarginalSamplerConditional} and the posterior sampler defined in \eqref{eq:preconditionedPosteriorSamplerPrior}--\eqref{eq:preconditionedPosteriorSamplerConditional}, as we demonstrate below.

To show that the diffusion defined in \eqref{eq:PreconditionedLangevin} would still target the right stationary measure, we define the following variable \cite{alrachid2018some}
\begin{align*}
    z = M^{-1/2} u.
\end{align*}
By applying Ito's formula, we obtain
\begin{align}
    \md Z_t &= -\nabla_z {\Phi}^M_\theta(Z_t) \md t + \sqrt{2} \md B_t, \label{eq:TransformedDiffusion}
\end{align}
where ${\Phi}^M_\theta(\cdot) = \Phi_\theta \circ M^{1/2}(\cdot)$. It is easy to see that, given $\Phi_\theta^M(z) = \Phi_\theta(M^{1/2} z)$, we can write
\begin{align*}
\nabla_z \Phi_\theta^M(z(u)) = D(u(z)) \nabla_u \Phi_\theta (u)
\end{align*}
where $D(u(z)) = M^{\top/2}$ is the Jacobian of the transformation. 


It can be also seen that the stationary measure of \eqref{eq:TransformedDiffusion}
\begin{align*}
    p_\infty(z) \propto \exp(-\Phi_\theta^M(z)),
\end{align*}
and the standard transformation of the random variables implies that this leaves the original stationary measure invariant, i.e., $p_\infty(u) \propto \exp(-\Phi_\theta(u))$ since $M$ is independent of $u$, the determinant term cancels.

It is easy to see that the convergence rate of the diffusion \eqref{eq:TransformedDiffusion} would be determined by the properties of the new map $\Phi_\theta^M$ for fixed $\theta$. Therefore, we first remark the properties of this map. In order to first analyse the sampler \eqref{eq:preconditionedMarginalSamplerPrior}--\eqref{eq:preconditionedMarginalSamplerConditional}, we first start with a lemma which follows from Assumption~\ref{assmp:LowerUpperBounds}.
\begin{lemma}\label{lem:PredLowerUpperBounds} There exists $m_M > 0$ and $L_M > 0$ such that
\begin{align}
0 < m_M = \inf_{\theta} \lambda_{\min} (M^{\top/2} A_\theta^\top G^{-1} A_\theta M^{1/2} )
\end{align}
and
\begin{align}
L_M = \sup_\theta \lambda_{\max}(M^{\top/2}  A_\theta^\top G^{-1} A_\theta M^{1/2}) < \infty.
\end{align}
\end{lemma}
The following remark then establishes the strong convexity and Lipschitz smoothness of $\Phi_\theta^M$.
\begin{remark}
We note that
\begin{align*}
    \|\nabla_z {\Phi}^M_\theta(z) - \nabla_z \Phi^M_\theta(z')\| \leq L_{\theta,M} \|z - z'\|,
\end{align*}
where $L_{\theta,M} = \lambda_{\max}(M^{\top/2} A_\theta^\top G^{-1} A_\theta M^{1/2})$. Also note that $\Phi^M_\theta$ is strongly convex with $m_{\theta,M}= \lambda_{\min}(M^{\top/2} A_\theta^\top G^{-1} A_\theta M^{1/2}) > 0$.
\end{remark}
\begin{proof}
  Assuming a valid preconditioning matrix (i.e. one that admits a square root $M =
  M^{1/2}M^{\top/2}$), then for any $x$ we have
  \begin{align*}
    x^\top M^{1/2} A_\theta^\top G^{-1} A_\theta M^{\top/2} x
    &= \left( M^{\top/2} x \right)^\top A_\theta^\top G^{-1} A_\theta M^{\top/2} x \\
    &= y^\top A_\theta^\top G^{-1} A_\theta y > 0
  \end{align*}
  where $y = M^{\top/2} x$. Positivity is guaranteed as $A_\theta^\top G^{-1}
  A_\theta$ is positive definite. Thus $m_M > 0$ as
  $\lambda_{\min}(A_\theta^\top G^{-1} A_\theta) > 0$ for any $\theta$.
\end{proof}
We define
\begin{align*}
    m_M = \inf_\theta m_{\theta, M} < \infty \quad\quad \textnormal{and} \quad \quad L_M = \sup_{\theta} L_{\theta, M} < \infty.
\end{align*}
We also define the worst-case condition number as in \eqref{eq:KappaMax} for the preconditioned case as
\begin{align}\label{eq:KappaMaxM}
    \kappa_{\max}^M = \frac{L_M}{m_M},
\end{align}
and note that $\kappa_{\max}^M < \infty$ as a corollary of Lemma~\ref{lem:PredLowerUpperBounds}.
Then, we can obtain the following bound. 
\begin{theorem}\label{thm:preconditionedConv} Assume $0 < \eta \leq \frac{m_M}{4 L_M^2}$. Then,
\begin{align*}
    \KL(p_k(u) || p(u)) \leq e^{-m_M \eta k} \bE\left[\KL(p_0(u), p(u|\theta))\right] + 8 \eta d  L_M \kappa^M_{\max},
\end{align*}
where $L_M$ and $\kappa_{\max}^M$ are defined in Lemma~\ref{lem:PredLowerUpperBounds} and Eq.~\eqref{eq:KappaMaxM} respectively.
\end{theorem}
\begin{proof}
Considering Eq.~\eqref{eq:TransformedDiffusion}, we consider the ULA for this diffusion as in the previous sections. This diffusion will converge to the measure $p(z|\theta)$, hence we obtain
\begin{align*}
    \KL(p_k(z | \theta) || p(z | \theta)) \leq e^{-m_{\theta,M} \eta k} \KL(p_0(z) || p(z|\theta)) + 8 \eta d \frac{L_{\theta, M}^2}{m_{\theta,M}}.
\end{align*}
Since the KL divergence is invariant under invertible transformations, we have $$\KL(p_k(M^{-1/2} u | \theta) || p(M^{-1/2}u | \theta)) = \KL(p_k(u | \theta) || p(u | \theta)),$$
which implies that we can obtain the following bound
\begin{align*}
    \KL(p_k(u | \theta) || p(u | \theta)) \leq e^{-m_{\theta,M} \eta k} \KL(p_0(u) || p(u|\theta)) + 8 \eta d \frac{L_{\theta, M}^2}{m_{\theta,M}}.
\end{align*}
Then, it is straightforward to arrive at the result by taking expectations and using the chain rule of KL-divergence.
\end{proof}
\begin{remark} Theorem~\ref{thm:preconditionedConv} shows that if one chooses a matrix $M$ such that
\begin{align*}
    \kappa_{\max}^M \ll \kappa_{\max},
\end{align*}
then a significant improvement can be made in terms of error guarantees. Also note that the step-size condition in Theorem~\ref{thm:preconditionedConv} (and similarly in other theorems) can be written as
\begin{align*}
    0 < \eta \leq \frac{1}{4 \kappa_{\max}^M L},
\end{align*}
replacing the condition $0 < \eta \leq \frac{1}{4 \kappa_{\max} L}$ appears in Theorem~\ref{thm:MarginalConvergence}. This means that, if $\kappa_{\max}^M \ll \kappa_{\max}$, one can choose much larger step-sizes under the preconditioned case. This is what we precisely observe in the experimental section below.
\end{remark}

Theorem~\ref{thm:preconditionedConv} straightforwardly extends to the posterior case, i.e., for the sampler defined in \eqref{eq:preconditionedPosteriorSamplerPrior}--\eqref{eq:preconditionedPosteriorSamplerConditional}. Similar to the case of the prior, we note the following lemma.
\begin{lemma}\label{lem:PrePostLowerUpperBounds} There exists $m^y_M > 0$ and $L^y_M > 0$ such that
\begin{align}
0 < m^y_M = \inf_{\theta} \lambda_{\min} (M^{\top/2} \nabla^2 \Phi_\theta^y M^{1/2} )
\end{align}
and
\begin{align}
L_M^y = \sup_\theta \lambda_{\max}(M^{\top/2}  \nabla^2 \Phi_\theta^y M^{1/2}) < \infty.
\end{align}
\end{lemma}
\begin{proof}
  The proof is similar to the proof of Lemma~\ref{lem:PredLowerUpperBounds} since the Hessian $\nabla^2\Phi_\theta^y$ is positive definite by Assumption~\ref{assmp:PosteriorLowerUpperBounds}.
\end{proof}
We define the worst-case condition number for the posterior case
\begin{align}\label{eq:KappaMaxMy}
    \kappa_{\max}^{M,y} = \frac{L_M^y}{m^y_M}.
\end{align}
The convergence result then follows in the identical way as we have derived for the prior. In other words, we consider the diffusion
\begin{align*}
    \md u_t = -M \nabla \Phi_\theta^y(u_t) \md t + \sqrt{2} M^{1/2} \md B_t,
\end{align*}
and proceed in the same way as in the prior, which we omit due to space concerns. We summarise the final result of our paper as follows.
\begin{theorem}\label{thm:preconditionedPostConv} Assume that Assumption~\ref{assmp:PosteriorLowerUpperBounds} holds and $0 < \eta \leq \frac{m_M^y}{4 {L^y_M}^2}$. Then,
\begin{align*}
    \KL(p_k(u|y) || p(u|y)) \leq e^{-m^y_M \eta k} \bE\left[\KL(p_0(u), p(u|\theta, y))\right] + 8 \eta d  L^y_M \kappa^{M,y}_{\max},
\end{align*}
where $L^y_M$ and $\kappa_{\max}^M$ are defined in Lemma~\ref{lem:PrePostLowerUpperBounds} and Eq.~\eqref{eq:KappaMaxMy} respectively.
\end{theorem}
\begin{proof}
The proof is identical to the proof of Theorem~\ref{thm:conditionalConv}, adapting the arguments for the prior to the posterior.
\end{proof}

\begin{figure}[t]
  \centering
  \includegraphics[width=0.6\textwidth]{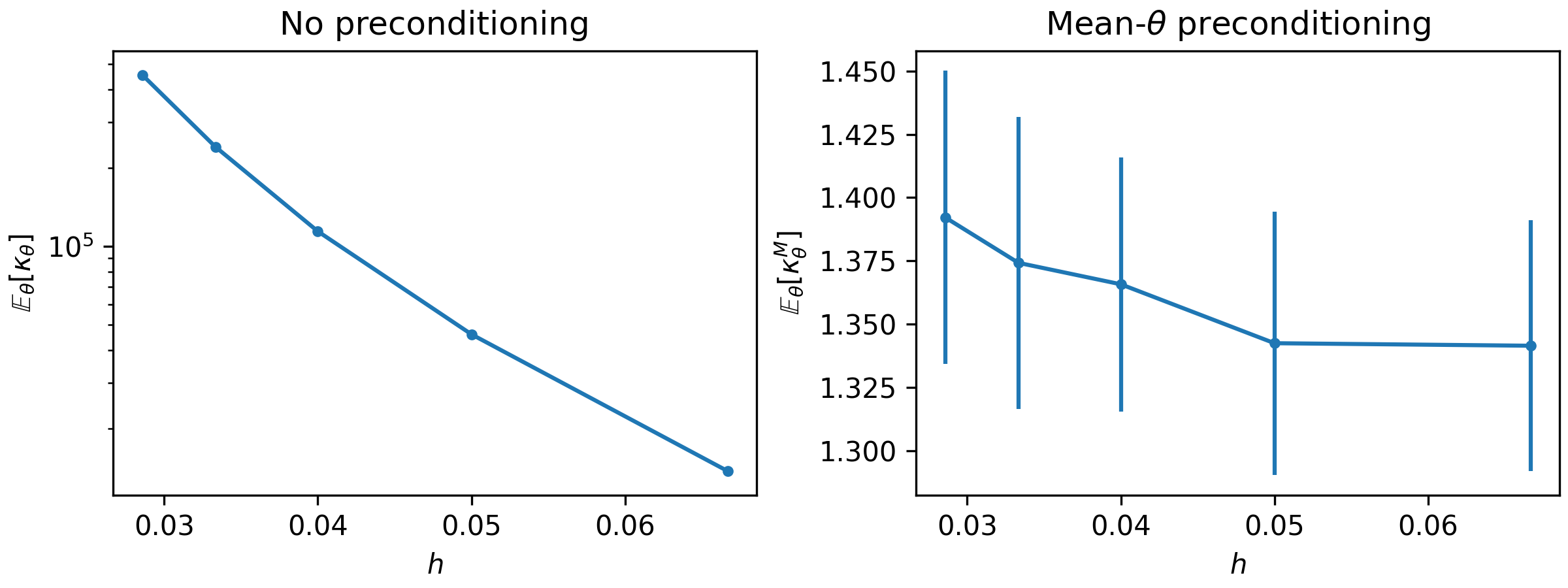}
  \caption{Expected condition number $\bE_{\theta}[\kappa_\theta]$ of the precision
    $A_\theta^{\top} G^{-1} A_\theta$ (left), and expected condition number
    $\bE_{\theta}[\kappa_\theta^M]$ of the preconditioned precision $M_{\bar{\theta}}^{1/2}
    A_\theta^{\top} G^{-1} A_\theta M_{\bar{\theta}}^{1/2}$ (right), for eight
    levels of mesh-refinement. Shown also are the estimated $50\%$ probability
    intervals.}
  \label{fig:prior-cov-conditioning}
\end{figure}

\begin{figure}[t]
  \centering
  \begin{subfigure}[t]{0.3\textwidth}
    \centering
    \includegraphics[width=\textwidth]{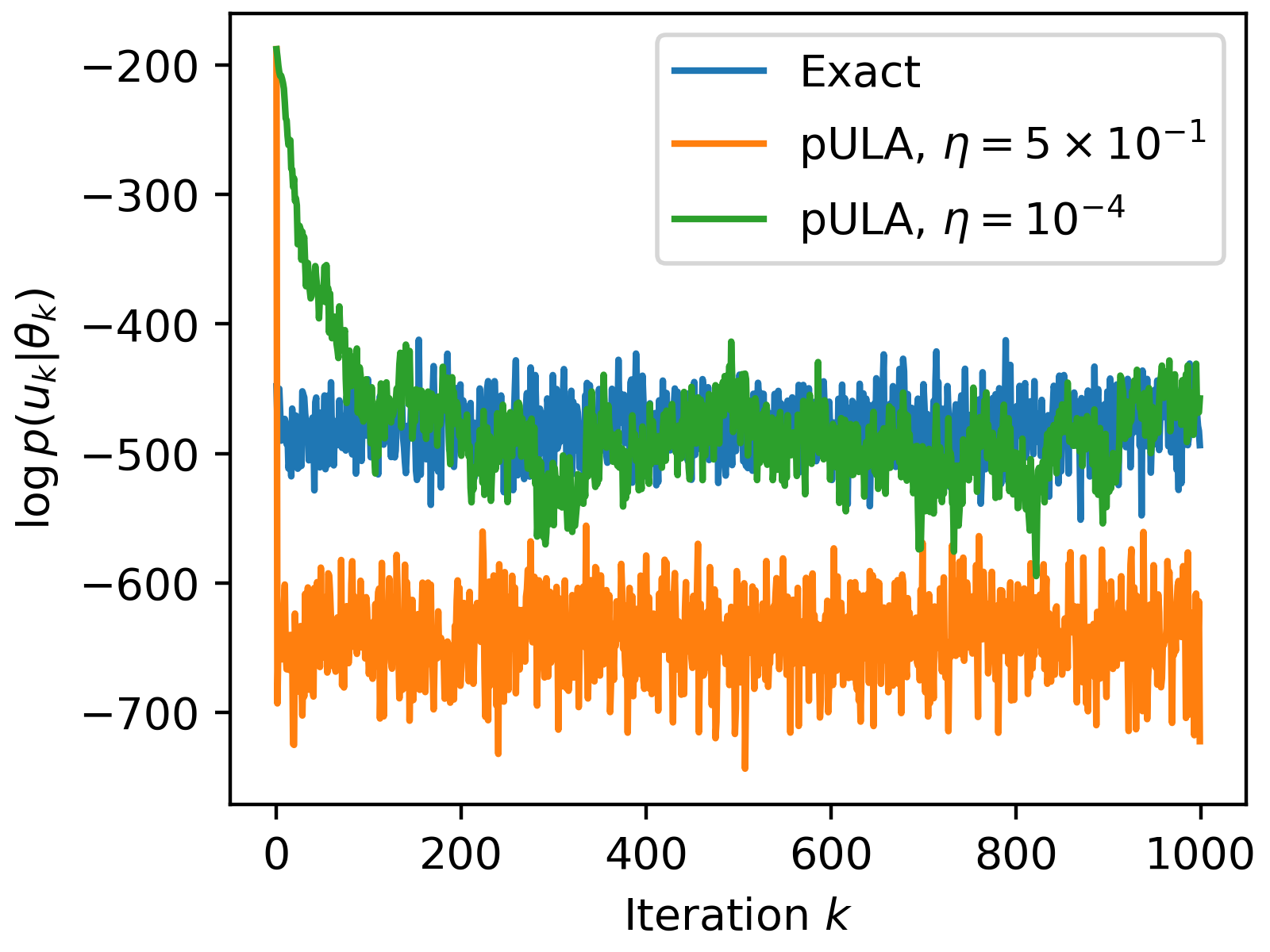}
    \caption{Trace plot.}
    \label{fig:small-prior-trace}
  \end{subfigure}
  \begin{subfigure}[t]{0.3\linewidth}
    \centering
    \includegraphics[width=\textwidth]{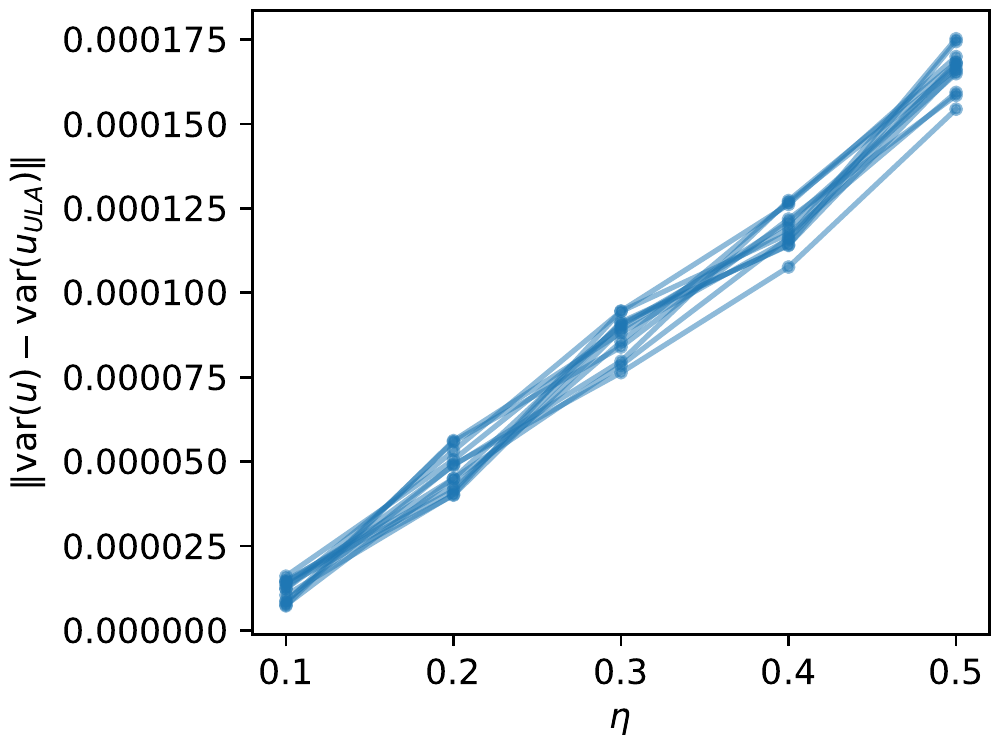}
    \caption{Variance error.}
    \label{fig:small-prior-var-error}
  \end{subfigure}
  \caption{Prior: results for state dimension $d = 1089$, plotting the sampled
    values of the log-target (a), and the errors on the variance, for $10$
    chains (b). For ULA, the stepsize offers a tradeoff between bias and rate of
    convergence: samplers that converge rapidly have a higher bias, which
    appears to increase linearly in $\eta$.}
  \label{fig:prior-trace-bias}
\end{figure}

\section{Empirical results}\label{sec:experiments}

In this section we illustrate the ULA methodology, both with and without
preconditioning, to sample from the prior and posterior measures. These results
illustrate that ULA performs comparably, in terms of mean and variance
estimates, to the standard Metropolis-adjusted algorithms. Preconditioning is
seen to improve convergence, and allows for larger stepsizes. The effect
of simulation parameters is also studied, quantifying the tradeoff between ULA bias, and
rapidity of convergence. All implementations are contained in an open-source
\texttt{Python} package, available
from~\url{https://www.github.com/connor-duffin/ula-statfem}. We use
\texttt{Fenics}~\cite{logg2012automated} for all finite element discretisations.

\subsection{Sampling the prior}

We consider the Poisson problem
\begin{align*}
  -\nabla \cdot \left( \theta(x) \nabla u(x) \right) &= f(x) + \xi(x), \quad x \in \Omega \\
  u &= 0, \quad x \in \partial \Omega
\end{align*}
where $\Omega = [0, 1] \times [0, 1]$ and $f(x) \equiv 1$. Both $\theta$ and $\xi$
are GPs, taking
\begin{gather*}
  \xi(x) \sim \mathcal{GP}(0, 0.05^2 \delta(x - x')), \\
  \log \theta(x) \sim \mathcal{GP}\left(\log\left(1 + 0.3 \sin(\pi(x_0 + x_1))\right),
    k_{\mathrm{se}}(x, x')\right), \;
  k_{\mathrm{se}} = 0.1^2 \exp\left(-\frac{\lVert x - x' \rVert^2}{2\cdot 0.2^2}\right).
\end{gather*}
Discretisation gives the conditional Gaussian $p(u | \theta) =
\mathcal{N}(A_\theta^{-1} b, A_\theta^{-1} G A_\theta^{-\top})$, which can be
marginalised over $\theta$ to give the prior
$p(u) = \int p(u|\theta) p(\theta) \, \md \theta$. We construct
$G$ by noting that
\[
  \begin{aligned}
    \tilde{G}_{ij} &= \beta^2 \int_\Omega \phi_i(x) \int \delta(x - x') \phi_j(x') \, \md x' \md x \\
    &= \beta^2 \int_\Omega \phi_i(x) \phi_j(x) \, \md x' \md x = \beta^2 M_{ij},
  \end{aligned}
\]
for the mass matrix $M$. A lumped approximation is made, to give the
diagonal $G_{ii} = \sum_j \tilde{G}_{ij}$.

We compare the available Langevin methods to sample from the marginal $p(u)$,
both with and without preconditioning. These are the Metropolis-adjusted Langevin
algorithm (MALA), preconditioned MALA (pMALA), ULA, and preconditioned ULA
(pULA). For preconditioning, the pMALA sampler uses the exact Hessian
$M = A_{\theta_k}^{-1} G A_{\theta_k}^{-\top}$, and the pULA sampler uses the
Hessian matrix with $\theta$ set to $\theta(x) = \bar{\theta}$, so $M =
A_{\bar{\theta}}^{-1} G A_{\bar{\theta}}^{-\top}$ (from here on in, this is
referred to as the mean-$\theta$-Hessian).
Figure~\ref{fig:prior-cov-conditioning} shows the effect of this preconditioner
on the prior covariance matrix, under mesh-refinement, using standard MC
sampling. We see that $\bE_\theta[\kappa_\theta^M]$ is approximately
six orders of magnitude smaller than $\bE_\theta[\kappa_\theta]$, being near
unity at various mesh-refinement levels. As the following examples will
show, this preconditioner is very effective; in it's absence the ULA sampler
may diverge due to instability, and MALA mixes poorly.

We first compare the results of running two pULA samplers on a low-dimensional
problem ($d = 1089$, triangular mesh with $32 \times 32$ cells), targeting
$p(u)$. Each sampler has a ``cold-start'', setting $u_0 = 0$, and stepsizes are
$\eta = 10^{-4}$ and $\eta = 5 \times 10^{-1}$. These results are shown in
Figure~\ref{fig:prior-trace-bias}, alongside samples from the exact measure.
With a higher stepsize the bias inherent in the ULA method can be seen, which is
traded for a faster converging sampler (Figure~\ref{fig:small-prior-trace}).
Variance errors are shown in Figure~\ref{fig:small-prior-var-error},
from running $5$ pULA samplers for $10$ chains each, with stepsizes increased
from $\eta = 0.1$ up to $\eta = 0.5$. The error in the variance increases
aproximately linearly as we increase the stepsize; for faster convergence, we
trade off more bias in the sampler.

Next, we increase the dimensionality, targeting $p(u)$ defined from the
stochastic Poisson problem with dimension $d = 16,641$ (FEM mesh with $128
\times 128$ cells). For each sampler we run a single chain with $5000$ warmup
samples, and $K = 10,000$ post-warmup samples. For the pULA sampler, each sample
uses $\ninner = 10$ inner iterations. All samplers are initialised with an exact
sample $u_0 \sim p(u)$; the warmup ensures that adequate
stepsizes are chosen. The MALA samplers were each run with an
adaptive stepsize during the warmup phase of the run, which increased or
decreased $\eta$ depending on the acceptance ratio, to ensure acceptance rates
of $\approx 50\%$. This gives the highly restrictive stepsize $\eta \approx 2.5
\times 10^{-10}$ for MALA, and $\eta \approx 6 \times 10^{-2}$ for pMALA. For
pULA, we use the dimensional stepsize of $d^{-1/3}$~\cite{roberts2001optimal} which is $\eta
\approx 3.9 \times 10^{-2}$. For a variety of stepsizes the ULA scheme diverges
due to numerical instability, typically failing in $<10$
iterations.\footnote{This is unsurprising given the poor condition number of the
covariance matrix, when considered with Theorem~\ref{thm:MarginalConvergence}.}
For this reason we do not show results for ULA, in this subsection.

\begin{figure}[t]
  \centering
  \begin{subfigure}[t]{0.3\textwidth}
    \includegraphics[width=\textwidth]{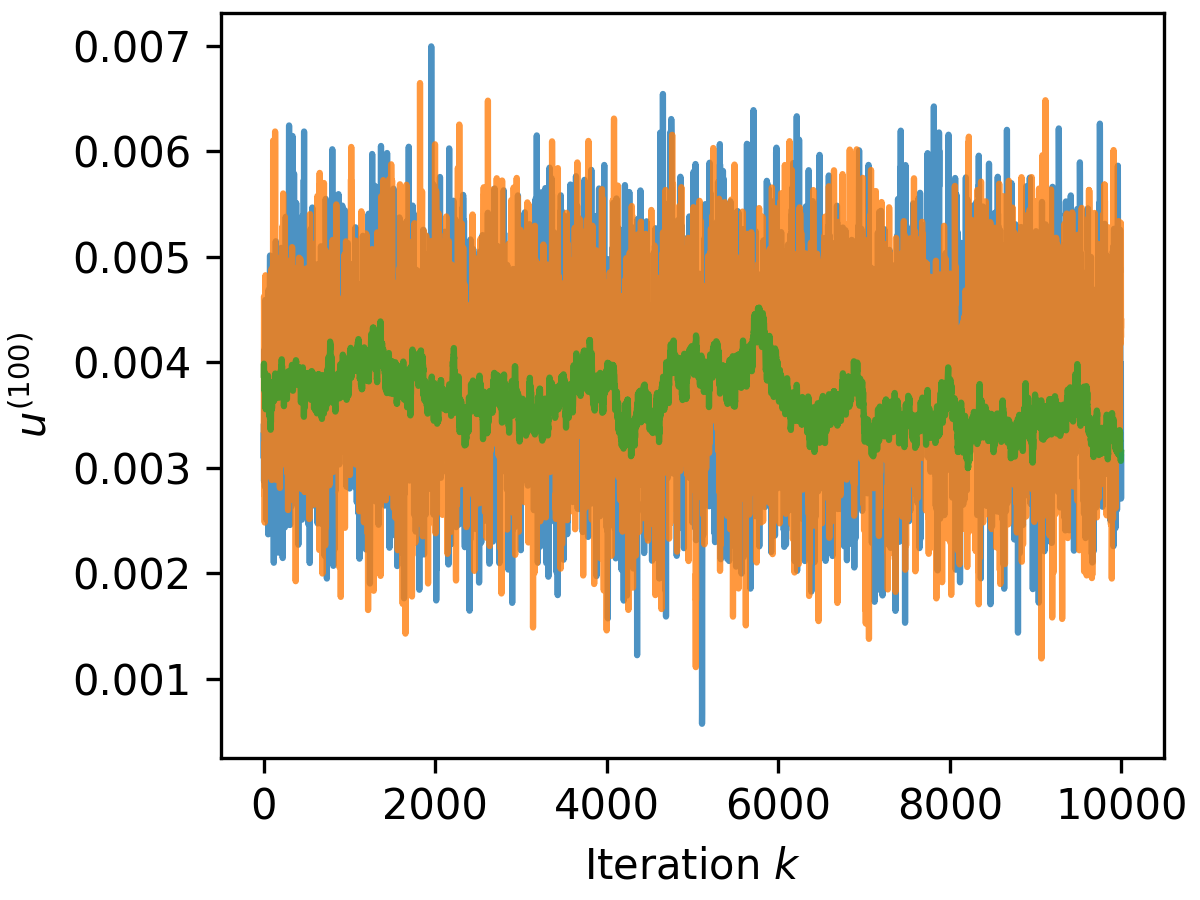}
    \caption{Trace plot.}
    \label{fig:prior-all-traceplot}
  \end{subfigure}
  \begin{subfigure}[t]{0.3\textwidth}
    \includegraphics[width=\textwidth]{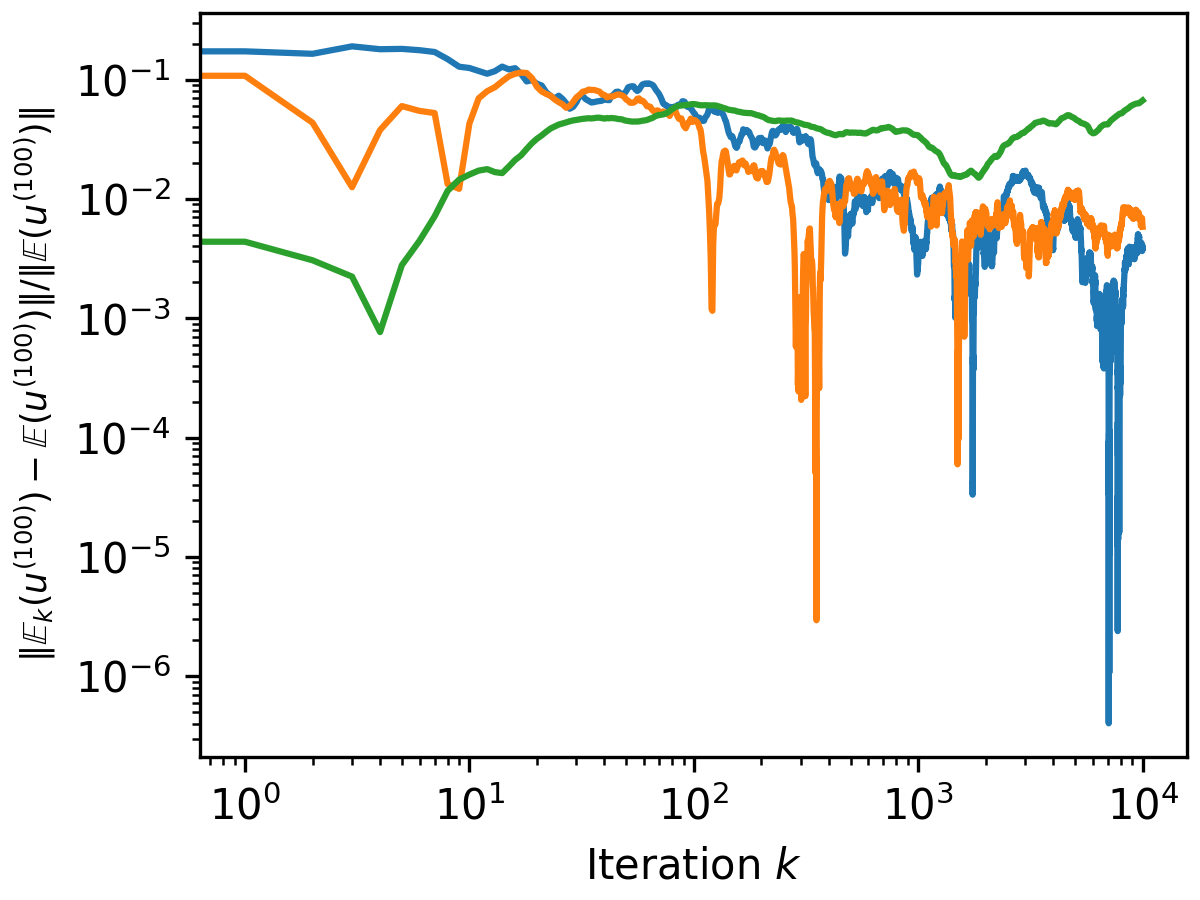}
    \caption{Relative error.}
    \label{fig:prior-dof-convergence}
  \end{subfigure}
  \begin{subfigure}[t]{0.3\textwidth}
    \includegraphics[width=\textwidth]{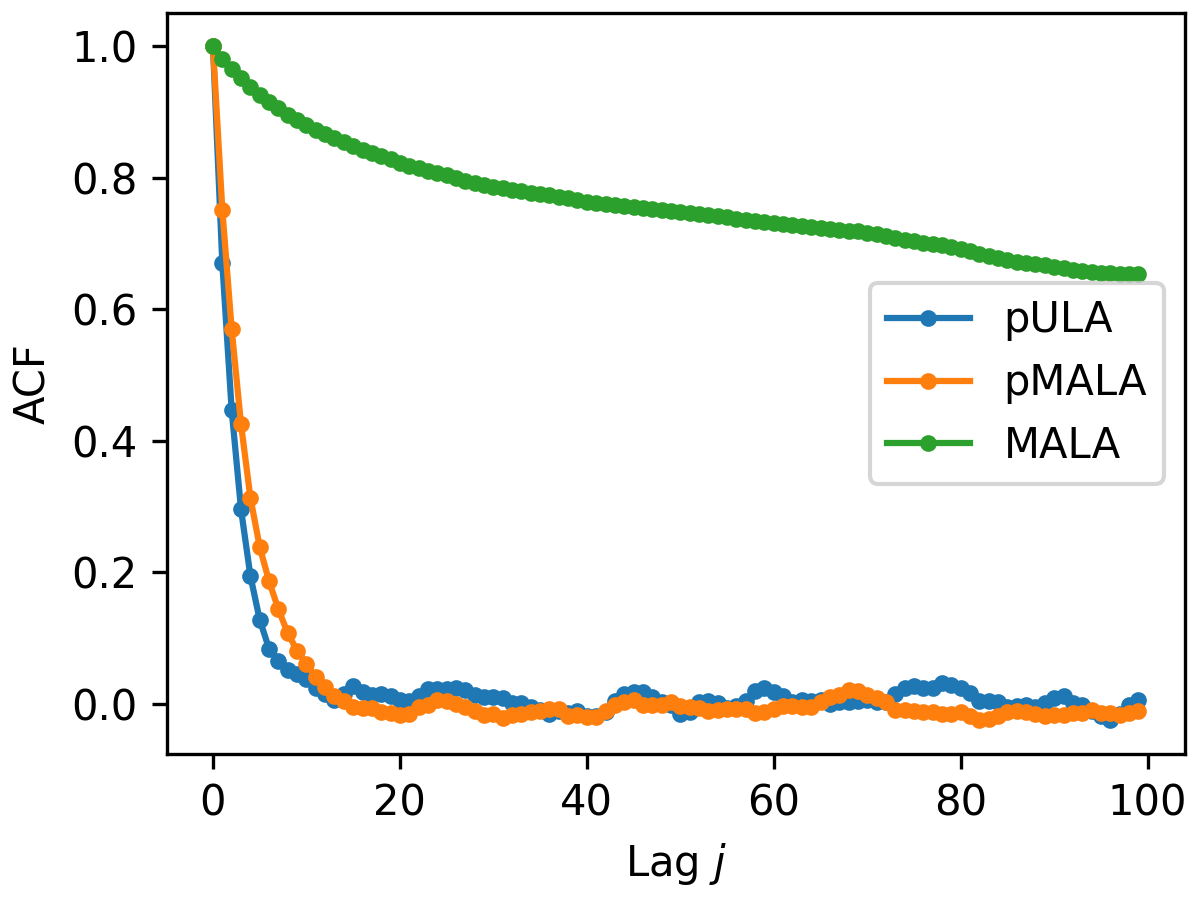}
    \caption{ACF plot.}
    \label{fig:prior-acf}
  \end{subfigure}
  \caption{
    Prior: results for state dimension $d = 16,649$. Shown are the results for
    the finite element coefficient $u^{(100)}$, using the post-warmup samples.
    Preconditioning is essential for this problem, to avoid poor mixing.
  }
  \label{fig:prior}
\end{figure}

Trace plots for the single FEM coefficient $u^{(100)}$ are shown in
Figure~\ref{fig:prior-all-traceplot} (other FEM coefficients give visually
similar plots), and demonstrate appropriate mixing. In this case the
exact-Hessian and the
mean-$\theta$-Hessian appear to give equivalent results (c.f. Table~\ref{tab:mean-var-errors}). The MALA sampler
fails to explore the space due to the poor conditioning
of $A_\theta^{-1} G A_\theta^{-\top}$. Figure~\ref{fig:prior-dof-convergence}
shows $\left| \bE_k(u^{(100)}) - \bE(u^{(100)})\right| / |\bE(u^{(100)})|$,
where $\bE_k(u^{(100)}) = \frac{1}{k} \sum_{i = 1}^k u_i^{(100)}$. The reference
$\bE(u^{(100)})$ is computed from $10,000$ exact samples from
$p(u)$. The convergence rates of the pULA and pMALA are visually equivalent, and the
MALA sampler records no decrease in the error. The autocorrelation functions
(ACFs), for the FEM coefficient $u^{(100)}$, are shown in
Figure~\ref{fig:prior-acf}, and accord with the previous results. The pULA and pMALA
samplers give similar performance, and the MALA samples appear highly correlated.

Next we check the accuracy of the mean and the variance estimates with the
relative errors of each, $\lVert \bE_K(u) - \bE(u) \rVert / \lVert \bE(u)
\rVert$, $\lVert \mathrm{var}_K(u) - \var(u)\rVert / \lVert
\var(u) \rVert$, computed against estimates from the exact reference
samples. These are displayed in Table~\ref{tab:mean-var-errors}, and show that
all samplers are accurate to order $O(10^{-2})$ in the mean. In the
variance, the pULA sampler is more accurate than the pMALA, and the MALA gives a
poor estimate. Note also that this provides no indication that the effects of
the ULA bias are practically realised for this example, with pULA having smaller
relative errors than the unbiased samples from pMALA. In terms of sampler efficiency, as
measured by the effective sample size per second (computed from the samples of
the FEM coefficient $u^{(100)}$), preconditioning results in one
and two orders of magnitude improvement over no preconditioning, for the pMALA
and pULA samplers, respectively.

\begin{table}
  \centering
  \begin{tabular}{lrrr}
\hline
 Sampler   &   $\mathsf{Error}(\mathbb{E}(u))$ &   $\mathsf{Error}(\mathrm{var}(u))$ &   ESS / s \\
\hline
 pULA      &                          0.004505 &                            0.038913 &     4.605 \\
 pMALA     &                          0.006356 &                            0.060541 &     0.745 \\
 MALA      &                          0.049279 &                            0.994318 &     0.019 \\
\hline
\end{tabular}
  \caption{Prior: relative norm errors of the mean and variance for each
    sampling method (errors are computed against the exact samples), for the $128
    \times 128$ mesh. The ESS is computed from samples of the FEM coefficient
    $u^{(100)}$.}
  \label{tab:mean-var-errors}
\end{table}

\subsection{Sampling the posterior}

We now detail the results for sampling the posterior $p(u \given y)$, using the
prior of the previous section $p(u \given \theta) = \mathcal{N}(A_\theta^{-1}b,
A_\theta^{-1} G A_\theta^{-\top})$. In this experiment we assume $n_\text{obs}$
measurements are obtained at $d_y$ locations. Denoting by $y_i \in \bR^{d_y}$,
gives the full dataset as $y = [y_1, \ldots, y_{\nobs}] \in \bR^{d_y \times \nobs}$.
The observation process is $y_i = H u + e$, $e \sim \NPDF(0, R)$, and
the likelihood is $p(y \given u) = \prod_{i = 1}^{\nobs} p(y_i \given u)$,
for $p(y_i \given u) = \NPDF(Hu, R)$.

Each observation vector is generated from an exact sample from the marginal
$u^i \sim p(u)$, interpolated at the observation locations. To induce some sort
of model mismatch, we scale this sample by a known scale factor, and then add on
noise $e^i \sim p(e) = \NPDF(0, R)$. We repeat this $\nobs$ times to give the
full dataset $y = [y_1, \ldots, y_{\nobs}]$.

The conditional posterior $p(u \given y, \theta)$ can be written out as
\begin{gather*}
  p(u \given \theta, y) = \mathcal{N}(m_{u, \theta}, C_{u, \theta}), \\
  m_{u, \theta} = C_{u, \theta} \left(
    A_\theta^\top G A_\theta u
    + \sum_{i = 1}^{n_\text{obs}} H^\top R^{-1} y_i
  \right), \\
C_{u, \theta}^{-1} = A_\theta^\top G A_\theta + n_\text{obs} H^\top R^{-1} H.
\end{gather*}
We compute the marginal posterior
$p(u \given y) = \int p(u\given y, \theta ) p(\theta) \, \mathrm{d} \theta$ with
MALA, pMALA, ULA, pULA, and pCN. Based on results in the previous section, the
preconditioned Langevin samplers (pULA and pMALA), employ
the mean-$\theta$-Hessian preconditioner, with
$M^{-1} = A_{\bar \theta}^\top G^{-1} A_{\bar \theta} + H^\top R^{-1} H$. The
proposal covariance matrix for pCN is the prior covariance. Full details are
contained in the appendix.

In this section, the state dimension is the same as previous, with
$d = 16,649$. We generate $\nobs = 100$ measurements of vectors of size
$d_y = 128$ to give the data $y$. For each of the Metropolis-based methods, we
generate $10,000$ samples as warmup for each sampler, and $K = 10,000$
post-warmup samples. For the ULA samplers, we generate $K = 10,000$ samples
using $10$ inner iterations for each sample. All samplers are initialised to a
sample from the exact posterior, $u_0 \sim p(u \given y)$. During the warmup
iterations, the adjusted MCMC algorithms MALA, pMALA, and pCN run an
adaptive stepsize algorithm which automatically tunes the stepsize to give
acceptance rates of approximately $50\%$. For the ULA sampler, we used
$\eta = 1 \times 10^{-9}$, and for the pULA sampler we used the stepsize
$\eta = d^{-1/3} \approx 3.9 \times 10^{-2}$.

\begin{figure}[t]
  \centering
  \begin{subfigure}[b]{0.4\textwidth}
    \includegraphics[width=\textwidth]{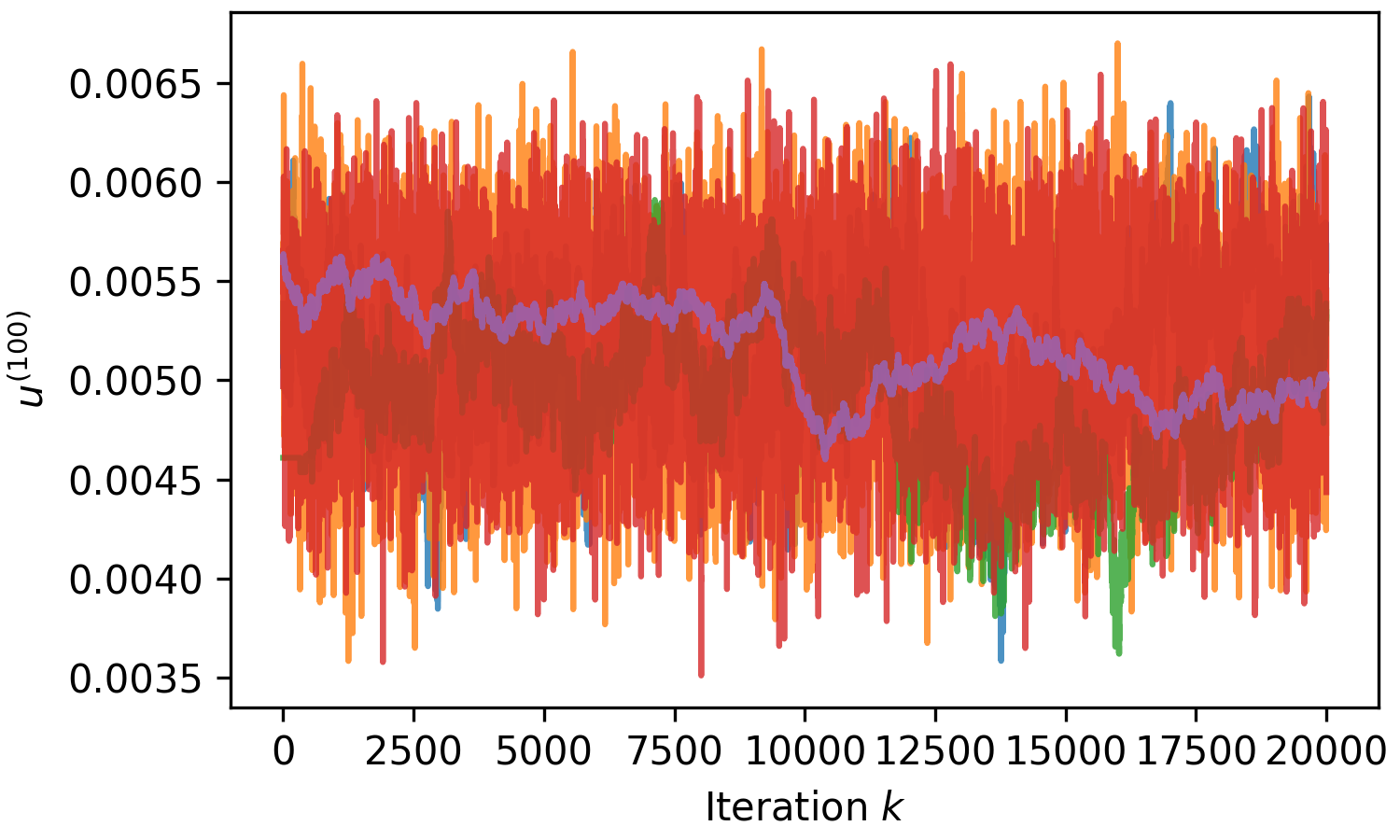}
    \caption{Trace plot.}
    \label{fig:post-traceplot}
  \end{subfigure}
  ~
  \begin{subfigure}[b]{0.32\textwidth}
    \includegraphics[width=\textwidth]{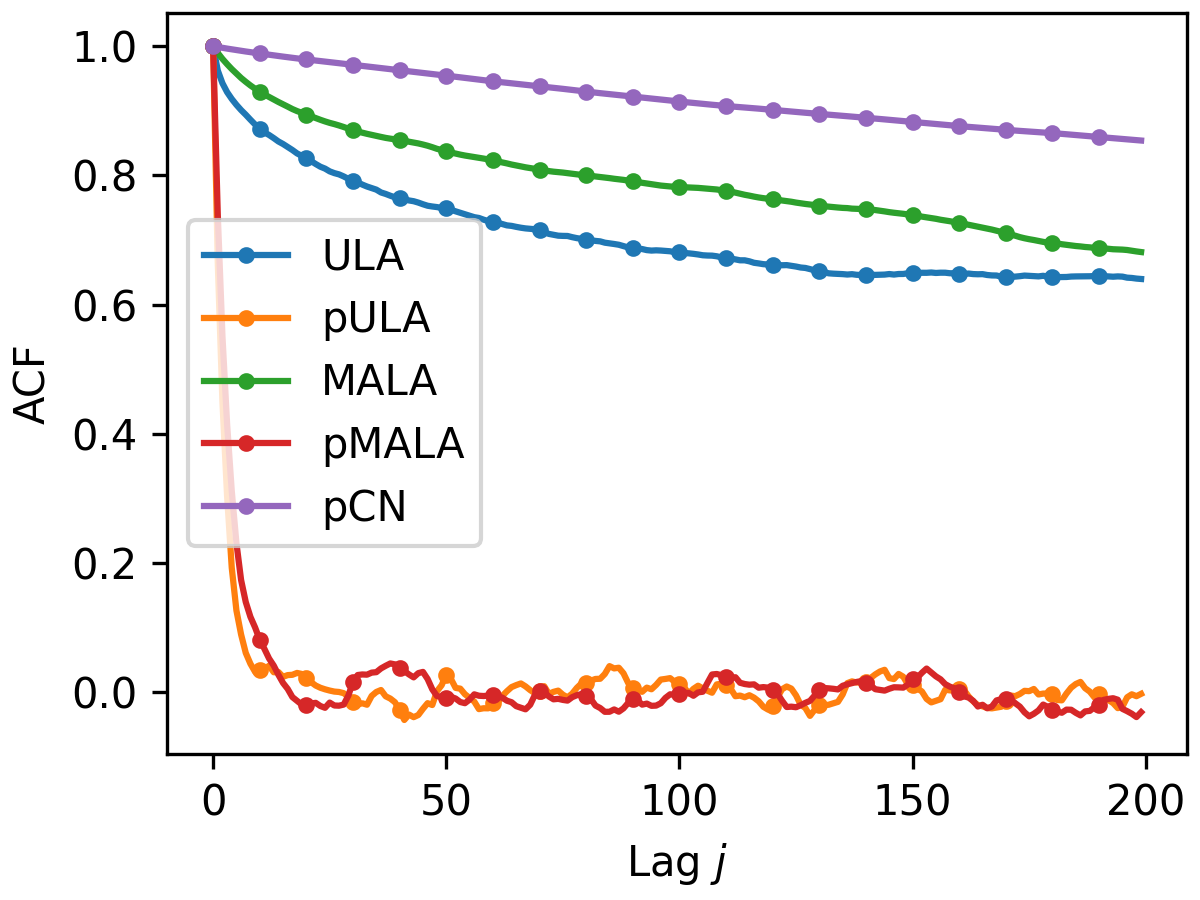}
    \caption{ACF plot.}
    \label{fig:post-acf}
  \end{subfigure}
  \caption{Linear posterior: results for state dimension $d = 16,649$; trace plots and
    ACF plots are shown for the FEM coefficient $u^{(100)}$. Compared to the
    standard pCN, the Langevin algorithms converge rapidly, with better
    UQ.}
\end{figure}

Trace plots for $u^{(100)}$ are shown in Figure~\ref{fig:post-traceplot} and
demonstrate that all samplers converge to the target measure, with the pCN
sampler giving notably poor mixing. This is also seen in the ACF
(Figure~\ref{fig:post-acf}), where pCN exhibits the highest correlation
correlated samples. For the non-preconditioned Langevin algorithms, adding in
the MH correction results in increased correlations. For the preconditioned
Langevin samplers, performance is comparable.

\begin{table}[t]
  \centering
  \begin{tabular}{lrrr}
\hline
 Sampler   &   $\mathsf{Error}(\mathbb{E}(u))$ &   $\mathsf{Error}(\mathrm{var}(u))$ &   ESS / s \\
\hline
 ULA       &                          0.005508 &                            0.546592 &     0.022 \\
 pULA      &                          0.000234 &                            0.025406 &     3.204 \\
 MALA      &                          0.008382 &                            0.785986 &     0.029 \\
 pMALA     &                          0.000257 &                            0.029731 &      1.52 \\
 pCN       &                           0.00814 &                            0.817542 &     0.012 \\
\hline
\end{tabular}
  \caption{Linear posterior: relative norm errors of the mean and variance, for each sampling
    method, and the ESS per second, for the post-warmup samples. The errors are
    computed against $10^4$ exact samples. The ESS/s is computed from the
    FEM coefficient $u^{(100)}$.}
  \label{tab:mean-var-errors-post}
\end{table}

Table~\ref{tab:mean-var-errors-post} shows the empirical results. The pCN, ULA,
and MALA perform similarly to one another: each is accurate in the mean, but is
inaccurate for the variance. Poor conditioning again makes the
non-preconditioned Langevin samplers mix poorly, and the pCN
struggles due to the prior covariance not capturing the posterior covariance
structure. Preconditioned Langevin samplers are more performant, with both pMALA
and pULA giving accurate estimates of the mean and variance, with two orders of
magnitude improvement in the sampler efficiencies over pCN. The pULA sampler is
approximately twice as efficient to the pMALA sampler. This is thought to be due
to pULA taking larger stepsizes, for which the increase in bias is not realised
in Table~\ref{tab:mean-var-errors-post} (errors on the variance are similar for
smaller for pULA vs. pMALA).

\subsection{Sampling the nonlinear posterior}

Finally, we detail the results for sampling the nonlinear posterior $p(u \given
y)$. In this case the data is generated from a nonlinear mapping of the FEM
coefficients, so that $y_i = \cH(u) + e$, $e \sim \NPDF(0, R)$,
for $i = 1, \ldots, \nobs$.
In this case, $\cH : \bR^d \to \bR^{n_y}$, and we assume that
$\cH(\cdot)$ is at least once differentiable. Data is simulated using the same
procedure of the previous subsection, except that where solutions are
interpolated with the observation operator $H$, we now use the nonlinear
$\cH(\cdot)$. In this case, we use a sigmoid function, defining
\[
  \cH(u)_j = S((Hu)_j), \; \text{ with } S(x) = \frac{0.1}{1 + e^{-100(x - 0.05)}},
\]
which has the interpretation that past input values of $\approx 0.1$, the
synthetic observation sensor saturates and is unable to distinguish between
input values.

The prior of the previous subsection is used,
$p(u \given \theta) = \NPDF(A_\theta^{-1} b, A_\theta^{-1} G A_\theta^{-\top})$,
so, up to an additive constant, the potential is
\[
  \Phi_\theta^y(u)
  = \frac{1}{2} \sum_{i = 1}^{n_{\text{obs}}} \left( y_i - \cH(u) \right)^{\top} R^{-1} \left( y_i - \cH(u) \right)
  + \frac{1}{2} (A_\theta u - b)^\top G^{-1} (A_\theta u - b).
\]
In this case no exact sampler is available so we use the same samplers of the
previous section (ULA, MALA, pULA, pMALA, and pCN). For the preconditioned
Langevin methods, the preconditioner $M$ uses the inverse Gauss-Newton
Hessian~\cite{chen2011hessian} defined at the maximum-a-posteriori (MAP)~\cite{murphy2012machine} estimate
$u_* = \argmin_{u \in \bR^d} \Phi_{\bar\theta}^y(u)$. Letting $\cJ(u) = \nabla_u \cH(u)$, this gives
$M_\theta^{-1} = A_\theta^\top G^{-1} A_\theta + \cJ(u_*)^\top R^{-1} \cJ(u_*)$.
In practice we also set $\theta = \bar{\theta}$, so that $M^{-1} = M_{\bar\theta}^{-1}$.

The experimental setup uses $\nobs = 100$ observations at $d_y = 128$ locations
across the mesh, so $y \in \bR^{d_y \times \nobs}$. Samplers are each
initialised to the MAP estimate, $u_0 = u_*$, and are run for $2 \times 10^4$
warmup iterations and $K = 2 \times 10^4$ sampling iterations. Stepsizes are the
same as previous. All samplers set $\ninner = 10$, apart from ULA, for which we
set $\ninner = 50$ (this was found to improve sampler performance with a mild
increase in runtime). The proposal covariance for pCN is again set to the prior
covariance matrix.

\begin{figure}[t]
  \centering
  \begin{subfigure}[c]{0.65\textwidth}
    \includegraphics[width=\textwidth]{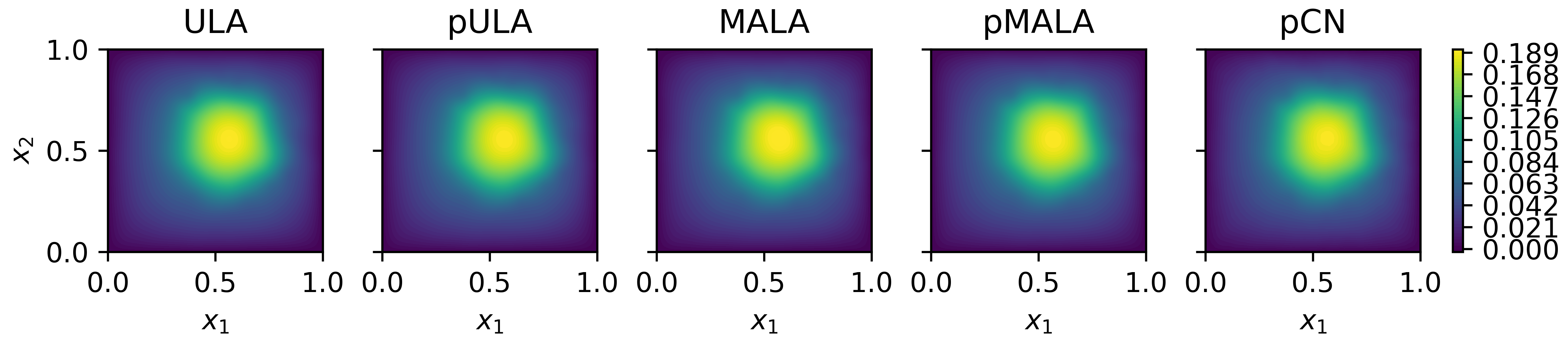}
    \includegraphics[width=\textwidth]{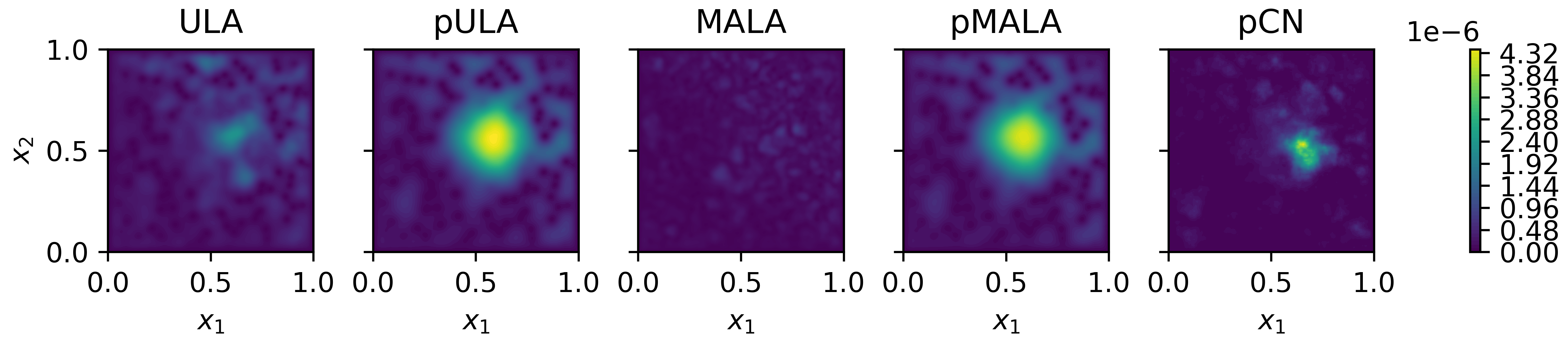}
    \caption{Posterior means ($\bE(u)$, top), and variance fields ($\var(u)$,
      bottom).}
    \label{fig:nll-means-vars}
  \end{subfigure}
  \begin{subfigure}[c]{0.3\textwidth}
    \includegraphics[width=\textwidth]{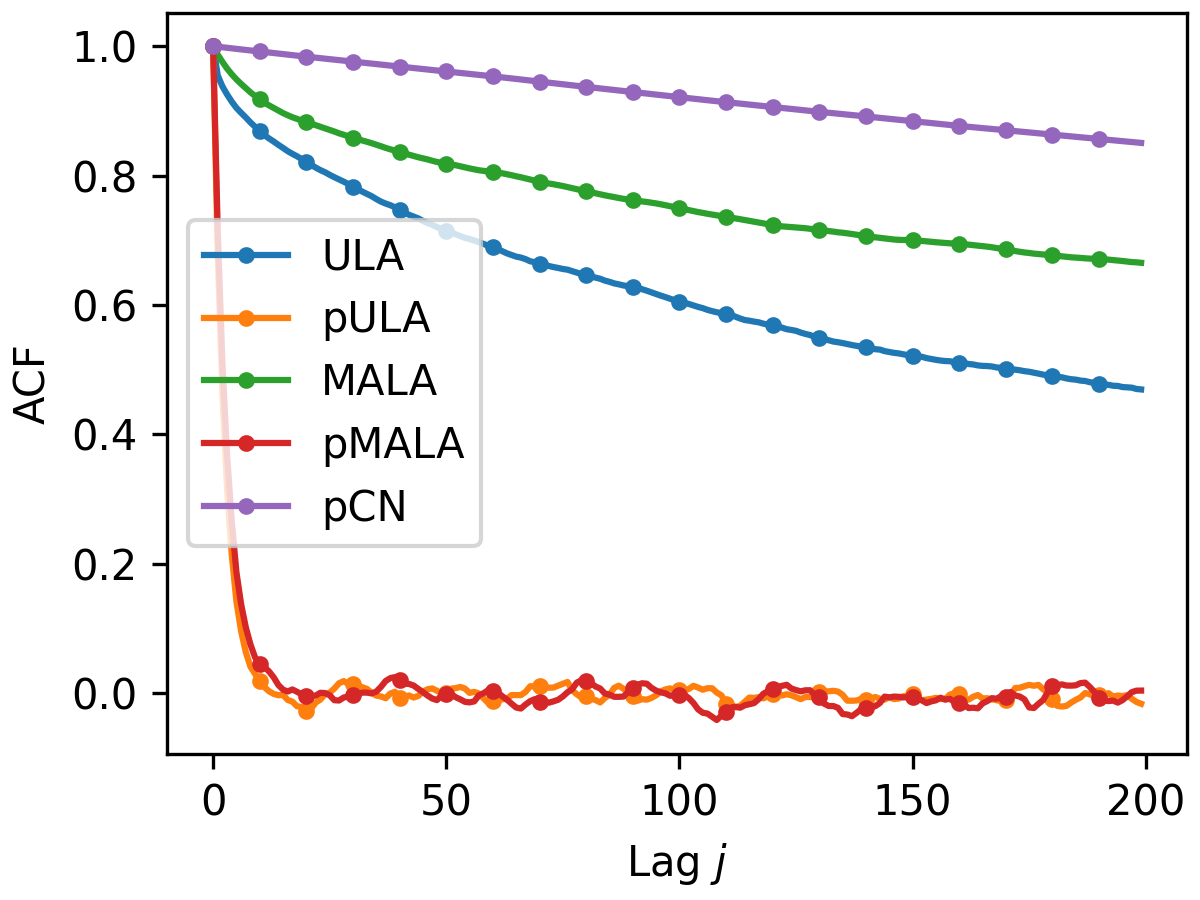}
    \caption{ACF plot.}
    \label{fig:nll-acfs}
  \end{subfigure}
  \caption{Nonlinear likelihood: posterior results. The ACF
    plot is shown for the samples from the FEM coefficient $u^{(100)}$. All
    samplers are accurate in the mean, but preconditioning captures the variance.}
  \label{fig:nll-results}
\end{figure}

Results are shown in Figure~\ref{fig:nll-results}. All posterior means show
agreement (Figure~\ref{fig:nll-means-vars}), as do the variance fields for pULA
and pMALA --- the other samplers struggle to provide accurate UQ (c.f.
Table~\ref{tab:nll-mean-var-errors}). The ACF,
which is shown in Figure~\ref{fig:nll-acfs}, suggests that correlations between
samples are similar to the linear posterior (c.f. Figure~\ref{fig:post-acf}),
with similar stratification of samplers. The Gauss-Newton Hessian, in this case,
is able to provide a sensible preconditioner that increases mixing, giving
a similar effect to that of the exact Hessian in the linear case.
Table~\ref{tab:nll-mean-var-errors} quantifies sampler performance, and gives
mostly similar results to the linear case. In this case MALA performs noticeably
worse due to the poorer conditioning of the nonlinear problem. ULA is able to
counteract this by taking larger stepsizes, whereas MALA remains constrained to
ensure optimal acceptance rates.

\begin{table}
  \centering
  \begin{tabular}{lrrr}
\hline
 Sampler   &   $\mathsf{Error}(\mathbb{E}(u))$ &   $\mathsf{Error}(\mathrm{var}(u))$ &   ESS / s \\
\hline
 ULA       &                          0.004063 &                            0.615644 &    0.0228 \\
 pULA      &                          0.000416 &                            0.050948 &    3.2322 \\
 MALA      &                          0.009769 &                            0.898006 &     0.002 \\
 pMALA     &                          0.000201 &                            0.014455 &    1.3233 \\
 pCN       &                          0.007586 &                            0.760288 &    0.0091 \\
\hline
\end{tabular}
  \caption{Nonlinear likelihood: relative norm errors of the mean and variance for each
    sampling method (errors are computed against $2 \times 10^5$ post-warmup
    pMALA samples), and the ESS is computed from the FEM coefficient $u^{(100)}$.}
  \label{tab:nll-mean-var-errors}
\end{table}

\section{Conclusions}
We have provided an iterative sampling approach for constructing finite element solutions to stochastic PDEs. Our approach is in the same vein as iterative FEM solvers and targets scenarios where a direct sampling approach might not be feasible. Furthermore, leveraging the standard tools of Bayesian inference, we have provided a principled way to incorporate data the statistical finite element model. This construction blends the data assimilation and the model construction, paving the way of customising these constructions for specific applications. The Langevin-based inference schemes we provide are generic as long as the gradients of the log-posteriors are available. Our future work plans include extending this framework to nonlinear PDEs using the results available for ULA for non-convex potentials \cite{zhang2019nonasymptotic, akyildiz2020nonasymptotic, lovas2020taming} and use well-studied probabilistic tools within this framework, e.g., model selection.
\label{sec:conclusion}

\section*{Acknowledgements}

We thank Valentin De Bortoli for useful discussions.

\bibliographystyle{siamplain}
\bibliography{draft}

\appendix

\begin{center}
  \section*{Appendices}
\end{center}

\small

\section{Lemmata}
\begin{lemma}\label{lem:ChainRuleKL} (The Chain Rule of the KL divergence) Let $p(x_1,\ldots,x_d)$ and $q(x_1,\ldots,x_d)$ be two arbitrary probability distributions on $\bR^d$. For $j \in \{1,\ldots, d\}$, we denote the marginals of $p$ and $q$ on the first $j$ coordinates by $p^{(j)}$ and $q^{(j)}$, i.e.,
\begin{align*}
  p^{(j)} = \int \cdots \int p(x_1, \ldots, x_d) \, \md x_{j+1} \ldots \md x_d,
\end{align*}
and
\begin{align*}
  q^{(j)} = \int \cdots \int q(x_1, \ldots, x_d) \, \md x_{j+1} \ldots x_d.
\end{align*}
Then
\begin{align*}
    \KL(p^{(1)} || q^{(1)}) \leq \KL(p^{(2)} || q^{(2)}) \leq \cdots \leq \KL(p || q),
\end{align*}
i.e., the function $j \mapsto \KL(p^{(j)} || q^{(j)})$ is increasing.
\end{lemma}

\begin{lemma}\label{lem:conditionalW2} Let $\mu_\theta$ and $\nu_\theta$ be two conditional measures. Let $\mu$ and $\nu$ be their marginals, i.e.
\begin{align*}
    \mu = \int \mu_\theta p(\md \theta) \quad \quad \textnormal{and} \quad \quad \nu = \int \nu_\theta p(\md\theta).
\end{align*}
Then, we have
\begin{align}
    W_2^2(\mu, \nu) \leq \bE_\theta \left[ W_2^2(\mu_\theta, \nu_\theta) \right].
\end{align}
\end{lemma}
\begin{proof}
Note that
\begin{align*}
    \bE_\theta \left[ W_2^2(\mu_\theta, \nu_\theta) \right]] &= \int \inf_{\gamma_\theta \in \Gamma(\mu_\theta, \nu_\theta)} \int \|u - u'\| \gamma_\theta(\md u, \md u') p(\md \theta).
\end{align*}
Now we take $\gamma^\star_\theta \in \Gamma(\mu_\theta, \nu_\theta)$ that attains the infimum and write
\begin{align*}
    \bE_\theta \left[ W_2^2(\mu_\theta, \nu_\theta) \right]] &= \int \int \|u - u'\| \gamma^\star_\theta(\md u, \md u') p(\md \theta), \\
    &= \int \|u - u'\| \gamma^\star(\md u, \md u'), \\
    &\geq W_2^2(\mu, \nu),
\end{align*}
where the second line follows from measurability of $\gamma_\theta^\star$ (see \cite[Corollary ~5.22]{villani2009optimal}).
\end{proof}

\section{Gaussian process on structured grids}

To speed up sampling $\theta(x)$ we make use of the regular mesh
\cite{saatci2011scalable}. In this work we use
two-dimensional regular meshes for all simulations. The log-Gaussian process
$\theta(x)$ can make use of this regular structure due to separability of
$k_\mathrm{se}(x, x')$ across spatial dimensions. We can write
the $n_m^2$ nodes of the mesh $\Omega_h$ as a Cartesian product between two sets
$\mathcal{X}_1 \times \mathcal{X}_2$, denoting the nodal locations as
$\mathcal{X}_1 = (x_{1,1}, \ldots x_{1, n_m})$, $\mathcal{X}_2 = (x_{2, 1},
\ldots, x_{2, n_m})$. Then the GP covariance matrix $K$ can be
written as a Kronecker product, $K = K_1 \otimes K_2$, where $K_{1, ij} =
k_\mathrm{se}(x_{1,i}, x_{1,j})$, $K_{2, ij} = k_\mathrm{se}(x_{2,i}, x_{2,j})$.
The Cholesky decomposition $K = LL^\top$ can also be written as
$LL^\top = \left( L_1 \otimes L_2 \right) \left( L_1 \otimes L_2 \right)^\top$,
where $K_1 = L_1L_1^\top$, $K_2 = L_2L_2^\top$.

This means that the covariance matrix never needs to be stored in memory.
Instead, the Kronecker factors $K_1$ and $K_2$ are stored, which gives
reduction in memory (in terms of floating point numbers) from
$\mathcal{O}((n_m^2)^2)$ to $\mathcal{O}(n_m^2)$.
Also, we need only take the Cholesky on the $n_m$-dimensional nodal locations,
which requires $\mathcal{O}(n_m^3 / 3)$ work), as opposed to
the full $2D$ mesh, which would require $\mathcal{O}((n_m^2)^3 / 3)$ operations.
Finally, making using of the ``vec trick'', we can perform matrix-vector
multiplications with $L = \left( L_1 \otimes L_2 \right)$ in $\mathcal{O}(n_m^2)$
time as opposed to $\mathcal{O}((n_m^2)^2)$.

\section{Metropolis-adjusted Langevin algorithms}

In this section we detail the Metropolis-adjusted Langevin samplers that we use
in this paper. For more details see \cite{girolami2011riemann}. For a fixed
$\theta$, recall that the ULA proposal step (as in
\eqref{eq:ULAiteration}), is given by
\[
  u_* = u_k - \eta \nabla \Phi_\theta(u_k) + \sqrt{2\eta} Z_{k+1},
\]
which defines a proposal density for a Metropolis sampler,
$q(u_* \given u_k, \theta) = \mathcal{N}(u_k - \eta \nabla \Phi_\theta(u_k), 2 \eta I)$.
Preconditioning the proposal with a symmetric positive-definite matrix $M$ gives
the proposal density
$q(u_* \given u_k, \theta) = \mathcal{N}(u_k - \eta M \nabla \Phi_\theta(u_k), 2 \eta M)$.
The acceptance ratio for the prior is
\[
  \alpha(u_*; u_k, \theta)
  = \min \left\{1, \frac{p(u_* \given \theta)}{p(u_k \given \theta)}
    \cdot \frac{q(u_k \given u_*, \theta)}{q(u_* \given u_k, \theta)}
  \right\},
\]
and is similarly defined for the posterior. This is computed on the log-scale to
avoid underflow.

For the joint update of $(u, \theta)$, we use a similar method
to that proposed for ULA: for a sampled value $\theta_k \sim p(\theta)$, we run
a MALA chain for a specified number of inner iterations $n_{\text{inner}}$,
sampling $u_{i}^{k} \sim p(u \given \theta_k)$, for $i = 1, \ldots,
n_{\text{inner}}$. The joint sample is then taken to be
$(u_{n_{\text{inner}}}^{k}, \theta_k) \sim p(u, \theta)$.

\section{Preconditioned-Crank Nicolson}

We also compare the preconditioned-Crank Nicolson (pCN) sampler in the posterior
case. The pCN sampler builds upon the idea that proposals which are reversible
with respect to the \textit{prior measure} only require computing a likelihood
ratio for the acceptance probability. In the linear case we sample from a posterior
$p(u \given y, \theta)$ which can be written as
\[
  p(u \given y, \theta) \propto \exp \left(
    -\frac{1}{2}  (y - Hu)^\top R^{-1} (y - Hu)
    - \frac{1}{2} (u - A_\theta^{-1} b)^\top \left( A_\theta^{-1} G
      A_\theta^{-\top} \right)^{-1} (u - A_\theta^{-1}b)
  \right).
\]
Now, if we consider the pCN proposal, defined by
\[
  u_* = A_\theta^{-1}b
  + \sqrt{1 - \beta^2} (u_k - A_\theta^{-1} b)
  + \beta \tilde{Z}_k,
  \quad \tilde{Z}_k \sim \mathcal{N}(0, A_\theta^{-1} G A_\theta^{-\top}),
\]
then by some algebra it can be shown that the acceptance probability is
\[
  \alpha(u_*, u_k; \theta) = \min \left\{
    1, \exp\left(
      -\frac{1}{2}  (y - H u_*)^\top R^{-1} (y - H u_*)
      +\frac{1}{2}  (y - H u_k)^\top R^{-1} (y - H u_k)
    \right)
  \right\},
\]
i.e., a likelihood ratio. For the joint update of $(u, \theta)$, again, we use a
similar method to that proposed for ULA: for a sampled value $\theta_k \sim
p(\theta)$, we run a pCN chain $n_{\text{inner}}$ iterations, sampling
$u_{i}^{k} \sim p(u \given \theta_k)$, for $i = 1, \ldots, n_{\text{inner}}$.
The joint sample is then taken to be
$(u_{n_{\text{inner}}}^{k}, \theta_k) \sim p(u, \theta)$.

\section{Numerical details}

We compute all samples using a workstation with an AMD Ryzen 9 5950X, with 128
GB memory. For the prior case, the exact sampler uses a smoothed aggregation
Algebraic Multigrid (AMG) solver, accelerated with Conjugate Gradients, with $4$
levels of coarseness to solve the linear system. This is implemented using the
\texttt{Python} package \texttt{pyAMG}~\cite{OlSc2018}. For the posterior case,
the exact sampler uses the sparse Cholesky decomposition as implemented in
CHOLMOD~\cite{chen2008algorithm}. The nonlinear likelihood is implemented in
JAX~\cite{jax2018github}, and makes use of automatic differentiation when computing
$\nabla_u \Phi_\theta^y(u)$. We also just-in-time compile (JIT) both the
log-likelihood and the gradient.

For applying preconditioners, in both the prior and posterior case, again we use
the sparse Cholesky decomposition (in this case the Hessian is sparse). For the
mean-$\theta$-Hessian preconditioners this is computed once and reused for the
entire chain. For the exact-Hessian preconditioner we compute the symbolic
factorization before running the chain, as the sparsity pattern of the Hessian
does not change.

\end{document}